\documentclass[reprint,superscriptaddress,amsmath,amssymb,aps,pra]{revtex4-1}
\usepackage[usenames,dvipsnames]{xcolor}
\usepackage{graphicx}
\usepackage{dcolumn}
\usepackage{bm}
\usepackage{braket}
\usepackage{booktabs, multirow}

\usepackage[colorlinks=true,urlcolor=blue,citecolor=blue,linkcolor=blue]{hyperref}
\usepackage{amsthm}

\newtheorem{lemma}{Lemma}
\newcommand{\proj}[1]{|#1\rangle\!\langle #1|}

\DeclareMathOperator{\CVaR}{CVaR}
\DeclareMathOperator{\CVaRv}{CVaRv}
\DeclareMathOperator{\Var}{Var}
\DeclareMathOperator{\tr}{tr}

\begin{document}

\title{Provable bounds for noise-free expectation values computed from noisy samples}

\author{Samantha V.~Barron}
\affiliation{%
IBM Quantum, IBM T.J.~Watson Research Center
}%

\author{Daniel J.~Egger}
\affiliation{%
IBM Quantum, IBM Research Europe - Zurich
}%

\author{Elijah Pelofske}
\affiliation{%
Los Alamos National Laboratory
}%

\author{Andreas B\"artschi}
\affiliation{%
Los Alamos National Laboratory
}%

\author{Stephan Eidenbenz}
\affiliation{%
Los Alamos National Laboratory
}%

\author{Matthis Lehmkuehler}
\affiliation{%
University of Basel
}%

\author{Stefan Woerner}
\email{wor@zurich.ibm.com}
\affiliation{%
IBM Quantum, IBM Research Europe - Zurich
}%

\date{\today}

\begin{abstract}
In this paper, we explore the impact of noise on quantum computing, particularly focusing on the challenges when sampling bit strings from noisy quantum computers as well as the implications for optimization and machine learning applications. We formally quantify the sampling overhead to extract good samples from noisy quantum computers and relate it to the layer fidelity, a metric to determine the performance of noisy quantum processors. Further, we show how this allows us to use the Conditional Value at Risk of noisy samples to determine provable bounds on noise-free expectation values. We discuss how to leverage these bounds for different algorithms and demonstrate our findings through experiments on a real quantum computer involving up to 127 qubits. The results show a strong alignment with theoretical predictions. 
\end{abstract}

\maketitle

\section{Introduction \label{sec:introduction}}

Quantum computing is a new computational paradigm which promises to impact many disciplines, ranging from quantum chemistry \cite{peruzzo_2014_vqe, ollitrault_2021_dynamics}, quantum physics \cite{dimeglio2023quantum}, and material sciences \cite{barkoutsos_2021_alchemical}, to machine learning \cite{Havlicek2019, Zoufal_2019_qgan, Zoufal_2021_varqbm}, optimization \cite{farhi_2014_qaoa, Bravyi2019, egger2021warm, Sack2023}, and finance \cite{Woerner_2019_risk, yndurain_2019_quantum_finance, Stamatopoulos_2022_market_risk}. 
However, leveraging near-term quantum computers is difficult due to the noise present in the systems. 
Ultimately, this needs to be addressed by quantum error correction, which exponentially suppresses errors by encoding logical qubits in multiple physical qubits \cite{nielsen_and_chuang, lidar_brun_2013_qec}. 

In near-term devices, implementing error correction is infeasible.
We must find other ways to handle the noise.
A promising approach to bridge the gap between noisy and error-corrected quantum computing is error mitigation.
Here, we leverage multiple noisy estimates to construct a better approximation of the noise-free result.
The most prominent examples are \emph{Probabilistic Error Cancellation} (PEC) \cite{berg2023probabilistic, Piveteau_2022} and \emph{Zero Noise Extrapolation} (ZNE) \cite{Temme_2017}. 
While error mitigation in general scales exponentially \cite{quek2023exponentially}, a combination of PEC and ZNE has been impressively demonstrated recently in a 127-qubit experiment at a circuit depth beyond the reach of exact classical methods \cite{kim_2023_utility, anand2023classical}.
The rate of the exponential cost of error mitigation directly relates to the errors in the quantum devices.
It is expected that these errors can be reduced to a level that noisy devices with error mitigation can already perform practically relevant tasks even before error correction \cite{Bravyi_2022}. 
PEC and ZNE mitigate the errors in expectation values.
While this finds many applications, e.g., in quantum chemistry and physics, most quantum optimization \cite{farhi_2014_qaoa, egger2021warm, zoufal_2023_blackbox} and many quantum machine learning algorithms \cite{Zoufal_2019_qgan, letcher2023tight} build directly on top of measured samples from a quantum computer. 
In optimization, having access to an objective value but not the samples corresponds to knowing the value of an optimal solution but not how to realize it.
Getting these samples is thus a key problem to scale sample-based algorithms on noisy hardware.

In this paper, we discuss the impact of noise on sampling bit strings from a noisy quantum computer and quantify the sampling overhead required to extract good solutions from noisy devices, e.g., in the context of optimization. 
Furthermore, we connect our findings to the \emph{Conditional Value at Risk} (CVaR, also known as Expected Shortfall), an alternative loss function introduced in Ref.~\cite{barkoutsos_2020_cvar}. 
We show that CVaR is robust against noise and can generate meaningful results from noisy samples also for expectation values.
This feature was already conjectured in Ref.~\cite{barkoutsos_2020_cvar} but has not been shown formally.
Our work closes this gap and shows that CVaR evaluated on noisy samples achieves provable bounds on noise-free observables. 
We demonstrate these bounds on up to 127-qubits on a real quantum computer applied to optimization problems, where we find close agreement between the experiments and theory.
In particular, this allows us to apply the known noise-free performance bounds for the \emph{Quantum Approximate Optimization Algorithm} (QAOA) for MAXCUT on 3-regular graphs \cite{farhi_2014_qaoa, wurtz_2021_qaoa}. Thus, our work thus results in \emph{provable} performance guarantees for a variational algorithm even on noisy hardware.

The remainder of this paper is organized as follows. 
First, Sec.~\ref{sec:sampling} discusses the impact of noise on sampling and how to quantify it.
Then, Sec.~\ref{sec:cvar} formally defines the CVaR and shows that it can provide provable bounds to noise-free expectation values from noisy samples.
Afterwards, Sec.~\ref{sec:applications} discuses the implications of the presented results in the context of applications in optimization, machine learning, and quantum time evolution. 
Sec.~\ref{sec:experiments} demonstrates the results on a real quantum computer up to 127-qubits where we find close agreement with the theory.
Last, Sec.~\ref{sec:conclusion} concludes the paper and we discuss open questions for further research.

\section{Sampling from Noisy Quantum Computers \label{sec:sampling}}

Suppose an initial $n$-qubit quantum state $\rho_0$, a quantum operation $\mathcal{U}(\cdot) = U \cdot U^{\dagger}$, and the resulting $\rho = \mathcal{U}(\rho_0)$. 
On a real quantum computer, we usually do not have access to the ideal operation $\mathcal{U}$ but only to a noisy version $\widetilde{\mathcal{U}}$ which we model by $\mathcal{U} \circ \Lambda$. 
Here, $\Lambda$ denotes the noise model. 
We denote the resulting noisy state by $\widetilde{\rho} = \widetilde{\mathcal{U}}(\rho_0)$.

For simplicity, we assume the Pauli-Lindblad noise model introduced in Ref.~\cite{berg2023probabilistic}
\begin{eqnarray}
    \Lambda(\rho) = \prod_{k \in \mathcal{K}} \left( w_k \, (\cdot) + (1-w_k) P_k (\cdot) P_k \right)\rho. \label{eq:noise_model}
\end{eqnarray}
Here, $\mathcal{K}$ denotes the index set for (local) Pauli error terms $P_k$, and $w_k = (1+e^{-2\lambda_k})/2$ for corresponding model coefficients $\lambda_k$ that determine the strength of the noise.
The assumption of Pauli noise can usually be justified via Pauli twirling \cite{knill_randomized_2008, dankert_exact_2009, magesan_scalable_2011}.
In Appendix~\ref{sec:beyond_pauli} we discuss Pauli twirling and the assumption of a Pauli noise model in more detail.

In general, a quantum circuit is not a single operation $\mathcal{U}$ but a concatenation of layers $\mathcal{U}_i$, $i=1, \ldots, l$. 
Their noisy versions are $\widetilde{\mathcal{U}}_i$ with corresponding noise models $\Lambda_i$. 
Crucially, this allows us to learn the noise model for each layer independently~\cite{berg2023probabilistic}.
A common assumption is that the layers $\mathcal{U}_i$ consist of non-overlapping CNOT gates (or other hardware-native two-qubit Clifford gates) and that these layers are possibly alternating with layers of single qubit gates.
Single qubit gates are assumed to be noise-free since their errors are an order of magnitude smaller than those of two-qubit gates.
Therefore, only the noise of the two-qubit gate layers is considered.

Assuming the above layer structure and that the noise model of the quantum processor is sparse allows Ref.~\cite{berg2023probabilistic} to introduce a protocol to efficiently learn the model coefficients $\lambda_k$. 
A property of $\Lambda$ that characterizes the overall strength of the noise is $\gamma = e^{2 \sum_k \lambda_k}$. 
This has a direct operational interpretation, since $\gamma^2$ defines the sampling overhead of applying PEC to mitigate the noise in the context of estimating an expectation value~\cite{Temme_2017, berg2023probabilistic}. 

Here, we first focus on sampling from noisy quantum computers instead of estimating expectation values.
Suppose we prepare a quantum state and afterwards measure the qubits. Then, the probability to sample a bit string $x \in \{0, 1\}^n$ is given by $p_x = \tr(\rho \proj{x})$ for the noise-free state $\rho$ and by $\widetilde{p}_x = \tr(\widetilde{\rho} \proj{x})$ for the noisy state $\widetilde{\rho}$.
The noise model introduced in Eq.~\eqref{eq:noise_model} can also be interpreted as follows: with a probability of $1/\sqrt{\gamma} = \prod_k w_k$ we sample a bit-string from $\rho$ and with probability $1 - 1/\sqrt{\gamma}$ we sample from a state where at least one error occurred.
Here, we assume $\lambda_k \ll 1$ such that we can leverage $e^x = 1+x + \mathcal{O}(x^2)$. It immediately follows that $w_k = e^{-\lambda_k} + \mathcal{O}(\lambda_k^2)$, and thus, $1/\sqrt{\gamma} = \prod_k w_k$.
Then, the law of total probability \cite{kokosaka_2000_probability} implies the lower bound:
\begin{eqnarray} \label{eq:prob_lower_bound}
    \widetilde{p}_x &\geq& p_x / \sqrt{\gamma}.
\end{eqnarray}
In other words, if a noise-free state $\rho$ has probability $p_x$ to sample a bit string of interest $x$, then, if $\rho$ is approximated by $\widetilde{\rho}$ prepared through a noisy process characterized by $\gamma$, we need a multiplicative sampling overhead of $\sqrt{\gamma}$ to guarantee at least the same probability of sampling $x$ as for the noise-free state. 
Thus, as long as we are only interested in generating relevant bit strings that we can efficiently evaluate classically, we can deal with the noise by measuring $\sqrt{\gamma}$-times more often.
This is in contrast to the multiplicative sampling overhead $\gamma^2$ introduced by PEC when we are interested in estimating expectation values.
Interestingly, if we apply PEC and then determine only the sampling probabilities, without evaluating an expectation value, we find that the sampling probabilities are lower bounded by $p_x / \gamma$, i.e., PEC ``amplifies'' the noise to achieve an unbiased estimation of expectation values, see Appendix~\ref{sec:sampling_pec} for more details.

The sampling overhead $\sqrt{\gamma}$ can be derived from the noise model resulting from the noise learning protocol introduced in Ref.~\cite{berg2023probabilistic}. 
However, in the present context, we are not interested in the full description of the noise model, only in $\gamma$.
Recently, Ref.~\cite{mckay2023benchmarking} introduced the \emph{Layer Fidelity} (LF), a metric to measure noise present in the hardware when executing a circuit. 
The LF also assumes the layered gate structure mentioned above and determines the resulting fidelity for each layer of gates.
It has a direct connection to the sampling overhead via $\text{LF}_i = 1 / \sqrt{\gamma_i}$, where $\gamma_i$ characterizes the noise of layer $i$. For multiple layers we can thus rewrite Eq.~\eqref{eq:prob_lower_bound} as 
\begin{eqnarray}
    \widetilde{p}_x &\geq& p_x \prod_i \text{LF}_i.
\end{eqnarray}
Further, the LF has the advantage that it is very cheap to evaluate compared to learning to full noise model. Thus, for a given circuit, the LF allows us to efficiently determine the sampling overhead to compensate the noise.

Other types of errors that we have not mentioned so far are \emph{state preparation and measurement} (SPAM) errors. 
In principle, we can also determine a sampling overhead and compensate for the SPAM errors by increasing the number of samples.
However, particularly for measurement errors, there exists other protocols which might allow for statistical corrections with a smaller sampling overhead \cite{van_den_Berg_2022_trex, Nation_2021_m3}. 
A systematic study of these types of errors would be interesting for future research.

\section{Conditional Value-at-Risk}\label{sec:cvar}

Section~\ref{sec:sampling} shows that we can sample bit strings of interest $x$, i.e., corresponding to the noise-free state $\rho$, by taking $\sqrt{\gamma}$-times more samples from the noisy state $\widetilde{\rho}$.
However, we usually do not know which samples correspond to the noise-free state and which samples were affected by noise.
We now leverage the insight of Sec.~\ref{sec:sampling} and show that the CVaR can provide provable bounds to noise-free expectation values from noisy samples.
The CVaR has already been suggested as a loss function and observable in Ref.~\cite{barkoutsos_2020_cvar}, however, only based on intuition and without theoretical justification.

Consider an integrable real-valued random variable $X$ with cumulative distribution function $F_X: \mathbb{R} \rightarrow [0, 1]$. Then, the (lower) CVaR at level $\alpha \in (0, 1]$ is defined as
\begin{align*}
    \CVaR_{\alpha}(X) &= \alpha^{-1} \mathbb{E}[X ; X \leq x_\alpha] \\
    &\qquad + x_\alpha (1 - \alpha^{-1}\mathbb{P}[X \leq x_\alpha ])\, ,
\end{align*}
where $x_\alpha = \inf\{x \in \mathbb{R} \colon F_X(x) \geq \alpha \}$. In the case when $F_X(x_\alpha)=\alpha$, this definition simplifies to $\CVaR_{\alpha}(X) = \mathbb{E}[X \mid X \leq x_\alpha]$, i.e. we are considering the expectation of $X$ when we are conditioning $X$ to take values in its bottom $\alpha$ quantile. Accordingly, we define the upper CVaR as 
\begin{eqnarray}
    \overline{\CVaR}_{\alpha}(X) = -\CVaR_{\alpha}(-X) \, .
\end{eqnarray}
Therefore we are considering the expectation of $X$ conditioned on values in its upper $\alpha$ quantile.
This allows us to prove the following lemma.

\begin{lemma}\label{lemma:cvar_bounds}
Suppose a random variable $X$ with probabilities $p_x = \mathbb{P}[X = x]$ for $x \in \mathbb{R}$. Further, suppose another random variable $\widetilde{X}$ as well as a given constant $C \geq 1$ such that $\widetilde{p}_x = \mathbb{P}[\widetilde{X} = x] \geq p_x / C$.
Then we have
\begin{eqnarray}
    \CVaR_{\alpha}(\widetilde{X}) &\leq \mathbb{E}[X] \leq& \overline{\CVaR}_{\alpha}(\widetilde{X}) \, ,
\end{eqnarray}
for all $\alpha \leq 1/C$.
Thus, the lower and upper CVaR of $\widetilde{X}$ with $\alpha \leq 1/C$ define lower and upper bounds, respectively, of the expectation value of $X$. 
\end{lemma}

\begin{proof}
By monotonicity of $\CVaR_{\alpha}(\widetilde{X})$ in $\alpha$, it suffices to show the claim for $\alpha=1/C$. Let $x_1<\cdots<x_n$ denote the support of $\widetilde{p}$. Take $k\leq n$ such that $\sum_{i\leq k-1} \widetilde{p}_{x_i} < 1/C \leq \sum_{i\leq k} \widetilde{p}_{x_i}$, then
\begin{align*}
    \CVaR_{1/C}(\widetilde{X}) = C \sum_{i\leq k} x_i \widetilde{p}_{x_i} + x_k \left(1 - C\sum_{i\leq k} \widetilde{p}_{x_i} \right) \,.
\end{align*}
Clearly, the $p$ minimizing $\mathbb{E}[X]=\sum_x x p_{x}$ and satisfying $p_x \leq C\widetilde{p}_x$ for all $x$ is also supported on $\{x_1,\ldots,x_n\}$ and satisfies
\begin{align*}
    p_{x_i} &= C\widetilde{p}_{x_i}\text{ for all $i<k$, and} \\
    p_{x_k} &\leq 1 - \sum_{i<k} p_{x_i} = 1 - C\sum_{i<k} \widetilde{p}_{x_i}
\end{align*}
From this, the claim is immediate by using the above to lower bound $\mathbb{E}[X]$. The upper bound follows by applying the lower bound to $-X$ and $-\widetilde{X}$ in place of $X$ and $\widetilde{X}$.
\end{proof}

Next, let us consider again a noise-free $n$-qubit quantum state $\rho$, its noisy version $\widetilde{\rho}$, and the corresponding $\gamma$. Further, suppose a diagonal Hamiltonian $H$, which can also be interpreted as a function $h: \{0, 1\}^n \rightarrow \mathbb{R}$.
Let us define the random variables $X, \widetilde{X} \in \{0, 1\}^n$, as the result of measuring $\rho$ and $\widetilde{\rho}$, respectively. 
Then, Lemma~\ref{lemma:cvar_bounds} and Eq.~\eqref{eq:prob_lower_bound} immediately imply
\begin{eqnarray} \label{eq:cvar_rho_bounds}
    \CVaR_{\alpha}(h(\widetilde{X})) &\leq \mathbb{E}[h(X)] \leq& \overline{\CVaR}_{\alpha}(h(\widetilde{X})) \, ,
\end{eqnarray}
for all $\alpha \leq 1/\sqrt{\gamma}$.
Since, for a diagonal $H$ we have $\tr(\rho H) = \mathbb{E}[h(X)]$, Eq.~\eqref{eq:cvar_rho_bounds} implies that the lower/upper CVaR computed from the noisy samples $\rho$ provide lower/upper bounds for the noise-free expectation value of $\rho$.
Further, suppose $\rho$ is the ground state of the diagonal $H$. 
Then, $h(\widetilde{X})$ cannot achieve any values smaller than $\tr(\rho H)$ and the left inequality in Eq.~\eqref{eq:cvar_rho_bounds} is an equality.
Thus, the noisy lower CVaR is equal to the ground state energy (similarly for the upper CVaR if $\rho$ would correspond to the maximally excited state of $H$).
Further, we also know that if the noisy CVaR would equal the ground state energy, the fidelity between the noise-free state $\rho$ and the noisy state $\widetilde{\rho}$ is lower bounded by the considered $\alpha$, i.e., $\mathcal{F}(\rho, \widetilde{\rho}) \geq \alpha$.

Diagonal Hamiltonians arise, e.g., in optimization problems or in the form of projectors $\proj{x}$, as can be used, e.g., for fidelity estimations. We will discuss these applications in more detail in Sec.~\ref{sec:qopt} and Sec.~\ref{sec:fidelities}.
However, many applications also involve non-diagonal Hamiltonians, most prominently applications in quantum chemistry and physics \cite{peruzzo_2014_vqe}.
Suppose a non-diagonal Hamiltonian $H = \sum_i c_i P_i$, where $P_i$ denote Pauli terms and $c_i$ the corresponding weights. Then, we can decompose $H$ into a sum of Hamiltonians consisting of subsets of commuting Pauli strings $H = \sum_j H_j$. 
All Pauli terms in $H_j$ can be simultaneously diagonalized via single qubit Pauli rotations. Thus, we can assume the $H_j$ are diagonal without loss of generality. 
We define the corresponding functions $h_j: \{0, 1\}^n \rightarrow \mathbb{R}$ as well as noise-free and noisy random variables $X_j, \widetilde{X}_j$, respectively, resulting from measuring the quantum states with the corresponding post-rotations to diagonalize the Hamiltonians $H_j$.
This implies
\begin{eqnarray}
    \sum_j \CVaR_{\alpha}(h_j(\widetilde{X}_j)) &\leq&
    \tr(\rho H) \nonumber \\ &\leq&
    \sum_j \overline{\CVaR}_{\alpha}(h_j(\widetilde{X}_j))\, ,
    \label{eq:cvar_non_diagonal_bounds}
\end{eqnarray}
for all $\alpha \leq 1/\sqrt{\gamma}$, which extends the previous result to non-diagonal Hamiltonians.
Note that in contrast to diagonal Hamiltonians, we cannot draw conclusions anymore about the groundstate energy or the fidelity between noisy state and groundstate. For instance, the lower bound in Eq.~\eqref{eq:cvar_non_diagonal_bounds} can be strictly smaller then the groundstate energy.

The CVaR can be estimated using Monte Carlo sampling. The variance of this estimator depends on the type of distribution considered but is always bounded by $\mathcal{O}(1/\alpha^2)$. 
However, for instance, for  Normal and Bernoulli distributions it can even be shown that in the present context the analytic behavior of the variances of CVaR for $\alpha \rightarrow 0$ is $\mathcal{O}(1/\alpha)$, where for Bernoulli, we assume that the success probability $p$ satisfies $p = \mathcal{O}(1/\sqrt{\gamma})$, which is the relevant case for the applications we consider later on, cf.~Sec.~\ref{sec:fidelities}. 
The derivation for the variance bounds for CVaR estimation are provided in Appendix~\ref{sec:cvar_variance}.
Thus, in these cases and for $\alpha = 1/\sqrt{\gamma}$, the variance increases as $\mathcal{O}(\sqrt{\gamma})$.
This renders the CVaR a very promising noise-robust loss function for variational quantum algorithms. The variance is amplified significantly less than for PEC, where it increases as $\mathcal{O}(\gamma^2)$. However, we need to recall that PEC comes with much stronger theoretical guarantees, i.e., provides an unbiased estimator instead of a bound. Thus, depending on the application, CVaR might not be applicable.

In the remainder of this section we discuss improvements to the lower and upper bounds for cases where we have more information about the noise-free state.
I.e, properties that the bit strings measured from the noise-free state must have but that might not persist under noise.
Examples of such properties are particle preservation in quantum chemistry~\cite{Bonet_Monroig_2018_post_selection, Choquette_2021} and constraints satisfaction in quantum optimization~\cite{barkoutsos_2020_cvar}.

Suppose a function $\mathcal{F}: \{0, 1\}^n \rightarrow \{0, 1\}$ that determines whether a bit string $x$ has a required property. 
Here, $\mathcal{F}(x) = 1$ indicates the presence of the property.
Further, suppose a given Hamiltonian $H$ and, for simplicity, let us assume it is diagonal and defined by a function $h: \{0, 1\}^n \rightarrow \mathbb{R}$. 
From this, we can construct a modified Hamiltonian $H_{\mathcal{F}}^{M}$ defined by the function
\begin{eqnarray}
    h_\mathcal{F}^{M}(x) &=& \begin{cases}
        h(x) & \text{if } \mathcal{F}(x) = 1, \\
        M & \text{otherwise,}
    \end{cases}
\end{eqnarray}
where $M$ is a given constant. 
We thus have $\tr(\rho H) = \tr(\rho H_{\mathcal{F}}^{M})$ in the noise-free case for any $M$, since all noise-free samples $x$ satisfy $\mathcal{F}(x) = 1$.
Next, we assume constants $M_l$ and $M_u$ that satisfy $M_l \leq h(x) \leq M_u$ for all $x$ with $\mathcal{F}(x) = 1$.
Samples with $\mathcal{F}(x) = 0$ must be affected by noise, which allows us to filter out samples where the noise destroys the required property.
Although there might still be noisy samples that are feasible, the post-selection reduces the impact of noise.
Due to the equality of expectation values in the noise-free case and the choice of $M_l$ and $M_u$, 
we immediately get
\begin{eqnarray}
\CVaR_{\alpha}(h_{\mathcal{F}}^{M_u}(\widetilde{X})) & \leq
\mathbb{E}[X] \leq & \overline{\CVaR}_{\alpha}(h_{\mathcal{F}}^{M_l}(\widetilde{X})),
\label{eq:post_selected_cvar_bounds}
\end{eqnarray}
for all $\alpha \leq 1 / \sqrt{\gamma}$.
This can lead to significantly better bounds since we can leverage the additional information about the considered problem to filter out more noisy samples. 
For non-diagonal Hamiltonians, see Eq.~\eqref{eq:cvar_non_diagonal_bounds}, it is possible to define a filter function $\mathcal{F}_j$ for each $H_j$.

Another implication of our results is that the average over the post-selected noisy samples must lie between the lower and upper bounds resulting from the filtered CVaR due to the monotonicity of CVaR with respect to $\alpha$.
Thus, the CVaR allows to bound the bias that post-selection may introduce and provide a quality measure for the estimated expectation value.

\section{Applications}\label{sec:applications}

We now discuss the presented theory on sampling probabilities and CVaR in the context of different applications: first, quantum optimization \cite{farhi_2014_qaoa, barkoutsos_2020_cvar, egger2021warm, zoufal_2023_blackbox, weidenfeller2022scaling}, and second, fidelity-based algorithms, such as Quantum Support Vector Machines (QSVM) \cite{Havlicek2019, gentinetta2022complexity, gentinetta2023quantum} as well as Variational Quantum Time Evolution (VarQTE) \cite{McArdle_2019_varqte, Yuan_2019_varqte, Zoufal_2021_varqbm, Zoufal_2023_varqte_error_bounds, Gacon_2021_qnspsa, gacon2023stochastic, gacon2023variational}. These are illustrative examples, the theory presented here is applicable to many other domains, such as quantum chemistry and physics.

\subsection{(Variational) Quantum Optimization \label{sec:qopt}}

Many variational quantum algorithms have been proposed to solve discrete optimization problems, such as Quadratic Unconstrained Binary Optimization (QUBO). 
Most of them have a similar structure and interpret every measured bit string as a potential solution to the problem.
Proposals that derive variable values from expectation values \cite{Bravyi2019, fuller2021approximate, teramoto2023quantumrelaxation, patti2022variational} are, however, not in the focus of our work.

Suppose a generic unconstrained binary optimization problem of the form
\begin{eqnarray}\label{eq:generic_binary_opt}
    \min_{x \in \{0, 1\}^n} f(x)\, ,
\end{eqnarray}
where $f: \{0, 1\}^n \mapsto \mathbb{R}$ is an objective function on $n$ binary variables. 
For instance, a QUBO has $f(x) = x^T Q x$ with $Q \in \mathbb{R}^{n \times n}$. 
In case of QUBO, we can apply a change of variables $x_i = (1 - z_i)/2$ for $z_i \in \{-1, +1\}$ and replace $z_i$ by the Pauli $Z_i$ matrix on qubit $i$ and products $z_i z_j$ by $Z_i \otimes Z_j$ to define a diagonal Hamiltonian $H$ and translate Eq.~\eqref{eq:generic_binary_opt} into a ground state problem~\cite{lucas_2014_ising}
\begin{eqnarray}
    \min_{\ket{\psi}} \braket{\psi | H | \psi}\, . \label{eq:ground_state_problem}
\end{eqnarray}
As mentioned in Sec.~\ref{sec:cvar}, we can transform any generic function to a Hamiltonian where $f(x)$ defines the diagonal element of $H$ at the position of the computational basis state $\ket{x}$~\cite{zoufal_2023_blackbox}.

Most variational quantum algorithms for binary optimization are defined via a parameterized ansatz $\ket{\psi(\theta)}$ with parameters $\theta \in \mathbb{R}^d$, a loss function $\mathcal{L}(\theta)$ that maps parameter values to a loss value, and an optimizer to solve
\begin{eqnarray}
    \min_{\theta \in \mathbb{R}^d} \mathcal{L}(\theta).
\end{eqnarray}
After the final parameters $\theta^*$ are determined, the resulting state $\ket{\psi(\theta^*)}$ is measured and the sampled bit strings are used as potential solutions to the problem. Samples obtained during the execution of the algorithm can also be considered as solutions in case they achieve better objective values than the final samples.

If we set $\mathcal{L}(\theta) = \braket{\psi(\theta) | H | \psi(\theta)}$ for some ansatz $\ket{\psi(\theta)}$, we get the \emph{Variational Quantum Eigensolver} (VQE) \cite{peruzzo_2014_vqe}.
Further, if we define the ansatz as
\begin{eqnarray}
    \ket{\psi(\theta)} &=& 
    \prod_{j=1}^p e^{-i H_X \beta_j} e^{-i H \gamma_j} \ket{+},
\end{eqnarray}
we get the QAOA \cite{farhi_2014_qaoa}, where $p$ defines the depth, $\beta_j, \gamma_j \in \mathbb{R}$ are the variational parameters, and $H_X = -\sum_{i=1}^n X_i$, where $X_i$ denotes the Pauli $X$ matrix on qubit $i$.

The results from Sec.~\ref{sec:sampling} and~\ref{sec:cvar} immediately apply to QAOA.
Suppose we already have a quantum circuit that, when executed and measured in an ideal noise-free setting, produces good solutions to a considered optimization problem.
Sec.~\ref{sec:sampling} immediately implies that when executed on a noisy devices, a sampling overhead of $\sqrt{\gamma}$ is sufficient to extract solutions of the same quality as in the noise-free case.
In certain cases it might be feasible to determine $\theta^*$ classically \cite{streif2019training, sack2021quantum} and only use the quantum computer to sample good solutions, since evaluating (local) expectation values might be easier than sampling from the full circuit \cite{begusic023simulating}.
However, in cases where we must train the parameterized quantum circuit we can replace the expectation value by the CVaR~\cite{barkoutsos_2020_cvar}.
The results introduced in Sec.~\ref{sec:cvar} now provide guidance on how to choose $\alpha$ and the required sampling overhead to get good results from a noisy device. 
We illustrate this on concrete examples in Sec.~\ref{sec:127_poly3} and Sec.~\ref{sec:maxcut_40_rr3}.

Our results allow us to apply proven performance guarantees for QAOA without noise to noisy hardware.
For MAXCUT on 3-regular graphs, QAOA achieves a worst-case performance of $0.692$ for $p=1$ \cite{farhi_2014_qaoa}, $0.7559$ for $p=2$, and (under certain assumptions) $0.7924$ for $p=3$ \cite{wurtz_2021_qaoa}. 
With a $\sqrt{\gamma}$ sampling overhead these guarantees are recovered even in the noisy regime. 
Furthermore, for 3-regular graphs, we can always train QAOA with $p \leq 3$ classically by simulating at most 30 qubits at a time \cite{Sack2023}, i.e., we can determine the optimal parameters via classical simulation and then sample good solutions with a $\sqrt{\gamma}$ overhead from the quantum computer.
Since $\gamma$ grows exponentially with the circuit size the sampling overhead introduced to combat noise may exceed the cost of a brute force search.
A simple back of the envelope calculation, discussed in Appendix~\ref{sec:brute_force}, determines a minimum layer fidelity require to apply a depth $p$ QAOA.

The \emph{Quantum Alternating Operator Ansatz} (QAOA') is an alternative of QAOA~\cite{hadfield_quantum_2019}. Here, a constraint, e.g., a fixed Hamming weight (i.e., a fixed number of ones in a bit string) is enforced by changing the mixer to preserve such states~\cite{wang2020xymixers,cook2020vertexcover,golden2023numerical} and starting in (a superposition of) feasible states~\cite{baertschi2022shortdepth,baertschi2020grover}. Thus, if QAOA' is executed noise-free, all resulting samples satisfy the given constraint. This is an example of a filter function $\mathcal{F}$, as introduced in Sec.~\ref{sec:cvar}, helps to improve the CVaR bounds on the corresponding expectation value.

\subsection{Fidelities\label{sec:fidelities}}

Several quantum algorithms leverage fidelity estimation between two quantum states in a sub-routine.
In the following, we first discuss how to leverage the CVaR bounds to approximate fidelities on noisy quantum computers and then how this impacts two concrete classes of algorithms: QSVMs and VarQTE.

Suppose we have $n$-qubit quantum circuits $U$ and $V$ that define $\ket{\psi} = U \ket{0}$ and $\ket{\phi} = V \ket{0}$, respectively. A common approach to estimate the fidelity between $\ket{\psi}$ and $\ket{\phi}$ is the \emph{compute-uncompute} method given by
\begin{eqnarray}
    \mathcal{F}(\ket{\psi}, \ket{\phi}) &=& \left| \braket{0 | V^{\dagger}U | 0} \right|^2.
\end{eqnarray}
$\mathcal{F}$ is thus the probability of measuring $\ket{0}$ for the state $V^{\dagger} U\ket{0}$.
This also equals the expectation value $\tr(\rho H)$ for the state $\rho = V^{\dagger} U\ket{0}$ and the diagonal Hamiltonian $H = \proj{0}$.
Thus, we can use $\overline{\CVaR}$ to get an upper bound of the noise-free fidelity.
Here, the resulting random variable follows a Bernoulli distribution, as the expectation value counts the number of measured $\ket{0}$'s and ignores all other outcomes.
Since the variance of the CVaR for a Bernoulli random variable scales with $1/\alpha$, see Sec.~\ref{sec:cvar}, we can set $\alpha = 1/\sqrt{\gamma}$ and use Eq.~\eqref{eq:cvar_rho_bounds} to upper bound the fidelity with a sampling overhead of $\sqrt{\gamma}$ compared to the $\gamma^2$ required by PEC to get an unbiased estimation.

QSVMs leverage a quantum feature map to define a quantum kernel and provably outperform classical computers on certain tasks~\cite{Liu_2021}.
The quantum feature map is a parameterized quantum circuit that takes a classical feature vector $x$ as an input to prepare a corresponding quantum state $\ket{\phi(x)}$.
The corresponding quantum kernel is then defined via the Hilbert-Schmidt inner product of $\ket{\phi(x_1)}$ and $\ket{\phi(x_2)}$ for two classical data points $x_1, x_2$ from some training set, which equals $\mathcal{F}(\ket{\psi}, \ket{\phi})$, and thus, falls exactly into the case above.

VarQTE for real or imaginary time evolution assumes a given parametrized quantum state $\ket{\psi(\theta)}$ and then projects the exact state evolution to the parameter evolution of the ansatz. This approximates the desired time evolution in the sub-space that the ansatz can represent. The exact projection requires the evaluation of the quantum geometric tensor (QGT)~\cite{McArdle_2019_varqte, Yuan_2019_varqte, Zoufal_2023_varqte_error_bounds}. 
However, that quickly becomes prohibitive as the number of parameters increases.
Thus, multiple approximate variants of VarQTE have been proposed that workaround the evaluation of the QGT \cite{Gacon_2021_qnspsa, gacon2023stochastic, gacon2023variational}. Many of these approximations leverage that the Hessian of the fidelity $|\braket{\psi(\theta) | \psi(\theta + \delta \theta)}|^2$ with respect to $\delta \theta$ which is proportional to the QGT of $\ket{\psi(\theta)}$ up to higher order terms.
They either use \emph{Simultaneous Perturbation Stochastic Approximation} (SPSA) to estimate the Hessian from evaluations of the fidelity as approximations of the QGT, or they construct alternative loss functions that directly leverage the mentioned fidelity without constructing an approximate QGT. In all variants, the parameter disturbances $\delta \theta$ are small, which implies fidelities close to one. Thus, this is in the regime where the noisy CVaR is very close to the noise-free expectation value, i.e., the sweet spot of the introduced approximation.

\section{Experiments\label{sec:experiments}}

Within this section, we analyze two optimization problems from the literature to demonstrate the theory presented in this paper.
In both cases, we run QAOA circuits on \emph{ibm\_sherbrooke}~\cite{ibm_quantum_devices}.
First, smaller but deeper circuits, and second, larger but more shallow circuits.
We always find a nice agreement between the theory and the experimental results.
All results within this section are achieved without twirling the circuits.
For a comparison and discussion of twirled and untwirled circuits see Appendix~\ref{sec:beyond_pauli}.

\emph{ibm\_sherbrooke} is a 127 qubit superconducting qubit device with an \emph{echoed cross-resonance} (ECR) gate as two-qubit gate~\cite{Sheldon2016}.
This gate is equivalent to a CNOT gate up to single-qubit gates and has a clear direction on the hardware.
We let the transpiler take care of the mapping from CNOT gates to ECR gates and will in the following write about CNOT gates for better readability.

\subsection{QAOA for MAXCUT on 3-regular graphs with 40 nodes}
\label{sec:maxcut_40_rr3}

\begin{figure}[t!]
    \centering
    \includegraphics[width=\columnwidth]{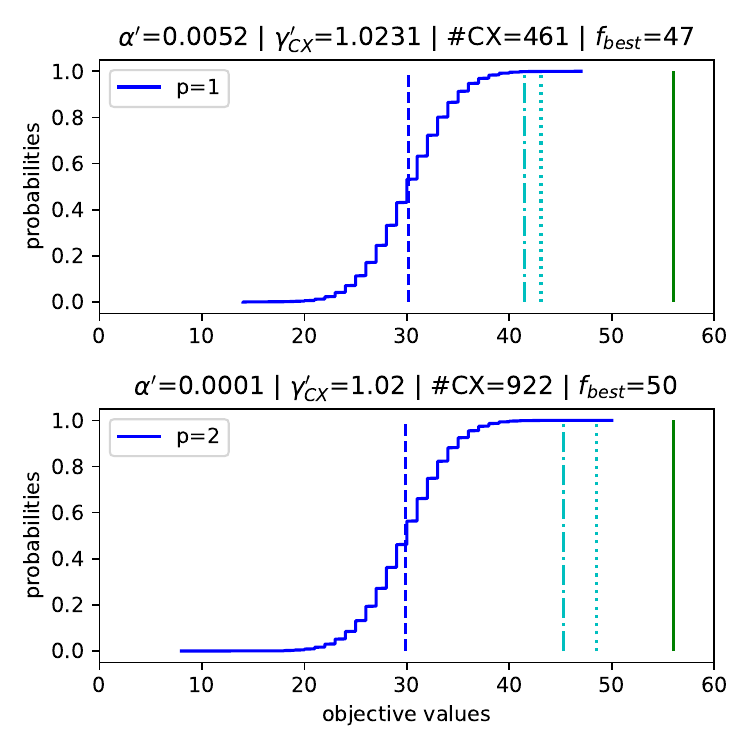}
    \caption{QAOA results on 40-qubits. 
    The curve is the cumulative distributions function resulting from sampling the circuits for a MAXCUT instance executed on \emph{ibm\_sherbrooke} for $p=1$ with $10^5$ shots (top) and $p=2$ with $10^7$ shots (bottom).
    The vertical lines show the corresponding noisy expectation values (dashed blue), the noise-free expectation values evaluated using light-cone optimized classical simulation (cyan dashed-dotted), the $\overline{\CVaR}_{\alpha_p}$ (cyan dotted), and the globally optimal solution equal to $56$ (green solid). The title shows the fitted $\alpha_p'$ such that the $\overline{\CVaR}_{\alpha_p'}$ are equal to the noise-free expectation values (i.e.~cyan dashed-dotted).}
    \label{fig:40_qubit_results}
\end{figure}

In this section, we examine QAOA for MAXCUT on a random three-regular graph with 40 nodes, i.e., on 40 qubits.
We take the problem instance from Ref.~\cite{Sack2023} and optimize the parameters classically for QAOA with depth $p=1$ and $p=2$ using light-cone simplifications.
This allows us to evaluate the required 2-local expectation values by simulating maximally 14 qubits at a time, see details in Ref.~\cite{Sack2023}.
The circuits and optimal parameters are further discussed in Appendix~\ref{sec:40_qubit_circuits}.

We apply staggered dynamic decoupling for error suppression, as discussed in Appendix~\ref{sec:staggered_dd}.
The circuits are constructed such that they consist of only two different layers of CNOT gates on a line of 40 qubits, denoted by $q_0, \ldots q_{39}$. The first layer is composed of 20 CNOT gates on qubits $(q_i, q_{i+1})$ for $i$ even and the second composed of 19 CNOT gates on $(q_i, q_{i+1})$ for $i$ odd.
Using the technique introduced in Ref.~\cite{mckay2023benchmarking} the measured LF for these two layers is $LF_1 = 0.7686$ and $LF_2 = 0.7444$, respectively~\footnote{At the time of writing the experiment to measure layer fidelity is under implementation in Qiskit Experiments~\cite{QiskitExperiments}. See \url{https://github.com/Qiskit-Extensions/qiskit-experiments}}.
We take the geometric average over the total number of CNOT gates and derive a CNOT fidelity as $\mathcal{F}_{CX} = (LF_1 \times LF_2)^{1/39} = 0.9858$. This also allows us to compute the \emph{error per layered gate} (EPLG) of Ref.~\cite{mckay2023benchmarking} as $1 - \mathcal{F}_{CX} = 0.0142$.
We also define $\gamma_{CX} = 1 / \mathcal{F}_{CX}^2 = 1.0290$.
In total, the circuits for $p=1$ and $p=2$ have 461 and 922 CNOT gates, respectively, all in form of the before mentioned layers. 
We can thus compute the sampling overhead for $p=1$ and $p=2$ as $\sqrt{\gamma_1} = 735.0$ and $\sqrt{\gamma_2} = 540275.9$, respectively, which corresponds to $\alpha_1 = 1.361\times 10^{-3}$ and $\alpha_2 = 1.851\times 10^{-6}$, for $p=1$ and $p=2$, respectively.
A regularly measured EPLG evaluated over a chain of 100-qubits is provided for \emph{ibm\_sherbrooke} in the IBM Quantum Platform~\cite{ibm_quantum_devices}.
At the time of the experiment the backend reported an EPLG of $0.017$, which is slightly higher than our measured EPLG.
This is expected, since we restrict to 40 qubits.
In any case, the EPLG reported by the backend is a good first proxy to estimate the LF and resulting $\gamma$ when executing a particular circuit on a device.

\begin{table}[t!]
    \centering
    \begin{tabular}{lcc}
                                    & $p=1$  & $p=2$ \\
        \hline
        global optimum              & \multicolumn{2}{c}{56} \\
        \hline
        $\mathbb{E}[\widetilde{X}]$ & 30.2 & 29.9 \\
        $\mathbb{E}[X]$             & 41.5 & 45.3 \\
        $\overline{\CVaR}_{\alpha_p}$    & 43.1 & 48.5 \\        
        best sampled value          & 47   & 50   \\
        \hline
        number of CNOT gates        & 461  & 922 \\      
        $\sqrt{\gamma_p}$           & $735.0$ & $540275.9$ \\
        $\alpha_p$                  & $1.361 \times 10^{-3}$ & $1.851\times10^{-6}$ \\
        $\alpha_p'$                 & $5.180 \times 10^{-3}$ & $1.071\times10^{-4}$ \\
        $\gamma_{CX}$               & \multicolumn{2}{c}{1.0290}  \\
        ${\gamma}_{CX,p}'$          & $1.0231$ & $1.0200$ \\
                
    \end{tabular}
    \caption{QAOA results on 40-qubits: This table shows the different results for $p=1$ and $p=2$ when running QAOA on the introduced 40-qubit MAXCUT instance. It shows the noisy and noise-free expectation values as well as the CVaR estimates, best sampled values and the global optimal value. Further, it shows the total number of CNOT gates, the overall $\sqrt{\gamma_p}$ for the circuits, the $\alpha_p$ derived from the LF as well as the $\alpha_p'$ derived from calibrating the CVaR on the noise-free expectation values, the corresponding $\gamma_{CX}$ and $\gamma_{CX,p}'$.}
    \label{tab:40_qubit_results}
\end{table}

To apply the CVaR bounds, we run the circuits for $p=1$ with $10^5$ shots and for $p=2$ with $10^7$ shots. 
This corresponds to 137 and 19 samples that remain to estimate the CVaR after sorting them and keeping the best $\alpha_1$ and $\alpha_2$ fraction, respectively.
The data confirm that $\overline{\CVaR}_{\alpha_p}$ provides an upper bound (since MAXCUT is a \emph{maximization} problem) to the noise-free expectation values, as predicted, see Fig.~\ref{fig:40_qubit_results} and Tab.~\ref{tab:40_qubit_results}.
The CVaR upper bound exceeds the noise free value by $3.9\%$ for $p=1$ and by $7.1\%$ for $p=2$.

We also use the noise-free expectation values obtained from the light-cone simulation to calibrate an $\alpha$ such that the CVaR matches the noise-free result exactly, denoted by $\alpha_p'$.
This allows us to derive an induced \emph{effective} $\gamma_{CX,p}'$ and compare it to the true $\gamma_{CX}$.
We find that $\gamma_{CX,p}'$ is quite stable for the different $p$ and significantly smaller than $\gamma_{CX}$, see Tab.~\ref{tab:40_qubit_results}. 
This may imply that the observable of interest is not affected by all the errors that may occur.
Crucially, this observation, may allow us to calibrate $\alpha$ for a particular application and choose larger values than implied by the LF, e.g., by running circuits of similar structure but with known noise-free results.
This may reduce the sampling overhead in certain scenarios while still achieving good results.
However, in general, the lower/upper bounds proven in Sec.~\ref{sec:cvar} will not hold anymore for $\alpha > 1/\sqrt{\gamma}$.

Comparing the $\overline{\CVaR}_{\alpha_p}$ and the best samples with the globally optimal solution, we find that they achieve approximation ratios of $0.770$ (CVaR) and $0.839$ (best sample) for $p=1$, and $0.866$ (CVaR) and $0.892$ (best sample) for $p=2$.
All these numbers exceed the corresponding theoretical lower bounds of $0.692$ ($p=1$) and $0.756$ ($p=2$) discussed in Sec.~\ref{sec:qopt}.

\subsection{QAOA on Hardware-efficient Higher-Order Ising Model with 127 variables
\label{sec:127_poly3}}

\begin{figure}
    \centering
    \includegraphics[width=0.99\columnwidth]{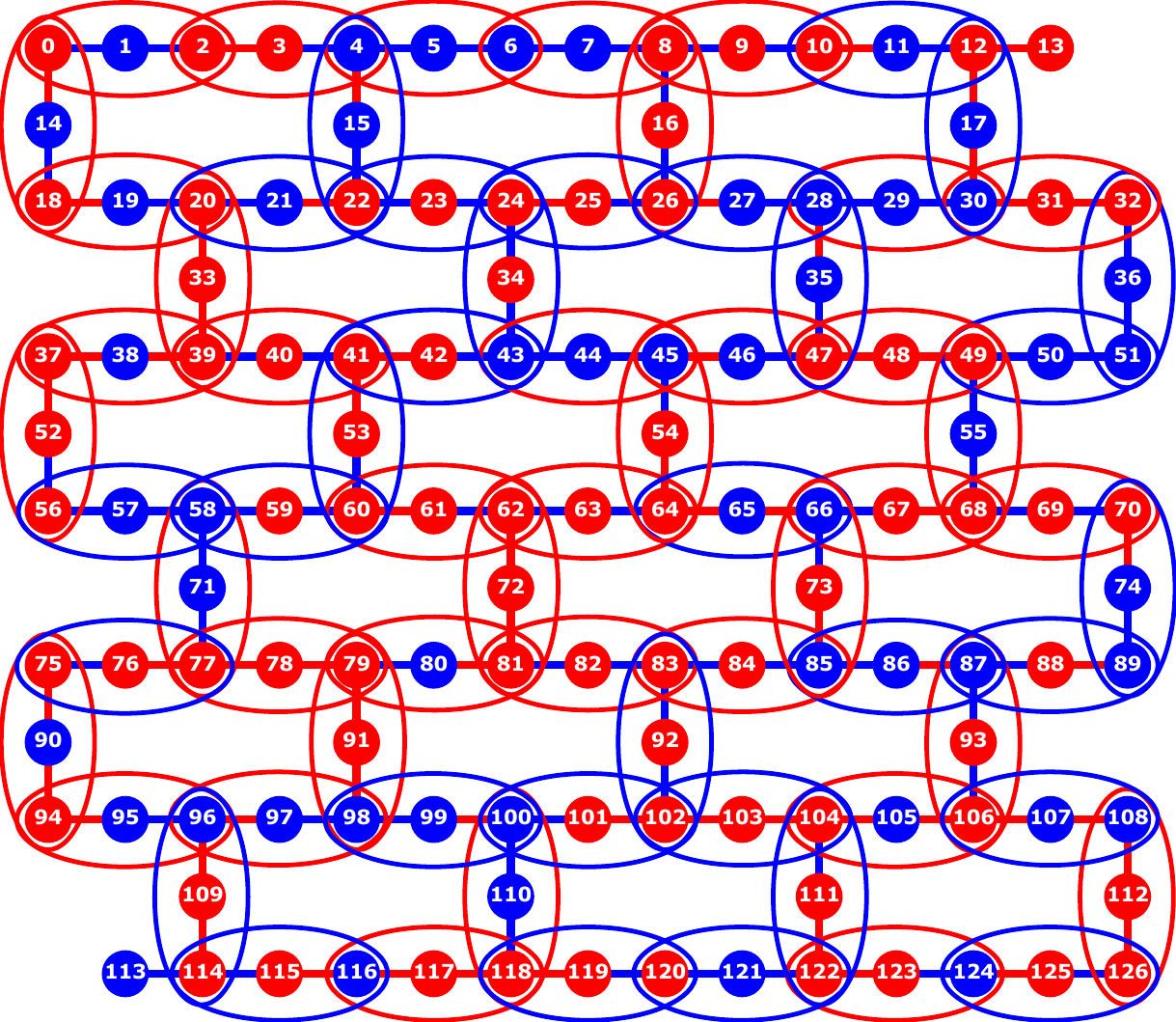}
    \caption{Example heavy-hex hardware compatible $127$ qubit higher order Ising model. Nodes denote the linear terms, edges between nodes denote the quadratic terms, and the ovals encircling three neighboring nodes on the hardware graph denote hyper-edges. Polynomial (Ising model) coefficients of $-1$ are denoted by red, and $+1$ are denoted by blue. }
    \label{fig:127_qubit_higher_order_Ising_model}
\end{figure}

We now show results of running QAOA on higher-order spin glass models. 
Originally described in Refs.~\cite{pelofske2023qavsqaoa, pelofske2023short}, these models are designed for a heavy-hex connectivity graph \cite{Chamberland_2020} of \emph{ibm\_sherbrooke}. 

We define a \emph{minimization} problem for the following cost Hamiltonian corresponding to a random coefficient spin glass problem with cubic terms and a connectivity graph that is defined to be compatible with an arbitrary heavy-hex lattice graph $G=(V,E)$, see Fig.~\ref{fig:127_qubit_higher_order_Ising_model}:
\begin{align}
        H =& \sum_{v \in V} d_v \cdot Z_v + \sum_{(i,j) \in E} d_{i,j} \cdot  Z_i \otimes Z_j  \nonumber \\
&    + \sum_{l \in W} d_{l,n_1(l),n_2(l)} \cdot Z_l \otimes Z_{n_1(l)} \otimes Z_{n_2(l)}.
    \label{equation:heavy_hex_higher_order_problem_instance}
\end{align}

As $G$ is a connected bipartite graph with vertices $V = \{0,\ldots,n-1\}$, it is uniquely \emph{bipartitioned} as $V=V_2 \sqcup V_3$ with $E \subset V_2 \times V_3$, where $V_i$ consists of vertices of degree \emph{at most}~$i$.
With~$W \subseteq V_2$ in~\eqref{equation:heavy_hex_higher_order_problem_instance}, we denote the subset of vertices in $V_2$ of degree \emph{exactly} $2$. Each node $l$ in $W$ has two neighbors, denoted by $n_1(l)$ and $n_2(l)$.
Thus $d_v$, $d_{i,j}$, and $d_{l,n_1(l),n_2(l)}$ are the coefficients representing the random selection of the linear, quadratic, and cubic coefficients, respectively. The random coefficients are chosen from $\{+1, -1\}$ with equal probability.
An example of such a random higher-order Ising model is in Fig.~\ref{fig:127_qubit_higher_order_Ising_model}. 

We use the qubits in $V_2$ to compute and uncompute parities into, for the $ZZ$ and $ZZZ$ terms in which they are contained. 
The unitaries $e^{-i\gamma ZZ}$ and $e^{-i\gamma ZZZ}$ are then realized with $R_z(2\gamma)$-rotations on these parity qubits. Computing and uncomputing parities needs $1+1$ and $2+2$ CNOT gates for the quadratic and cubic terms, respectively; however the CNOT gates for $Z_{l}Z_{n_1(l)}$ and $Z_{l}Z_{n_2(l)}$ can be subsumed into the CNOT gates for $Z_l Z_{n_1(l)} Z_{n_2(l)}$. 

Furthermore, $G$ as a bipartite graph of maximum degree 3 admits a 3-edge-coloring due to K\H{o}nig's line coloring theorem, meaning that these $2+2$ CNOT gates can be scheduled simultaneously for all terms in just $3+3$ non-overlapping layers~\cite{pelofske2023qavsqaoa}.
Depth-$p$ QAOA circuits for these problems thus have a CNOT depth of only $6p$, independent of the system size $n$. 
Further circuit details are given in Appendix~\ref{sec:127_qubit_circuits}.

Leveraging parameter transfer of QAOA angles for problems with the same structure but varying numbers of qubits, allows us to obtain good angles for these $127$ qubit QAOA circuits for $p=1, \ldots, 5$, without on-device variational learning \cite{heavy_hex_QAOA_parameter_transfer2023}. Additionally, we utilize converged MPS simulations with a bond dimension of $\chi=2048$ to verify that the fixed QAOA angles produce good expectation values \cite{heavy_hex_QAOA_parameter_transfer2023}, for all circuits. The hardware-compatible circuits are run on the \emph{ibm\_sherbrooke} device, again using staggered dynamic decoupling for error suppression, see Appendix~\ref{sec:staggered_dd}. The optimal solutions of the higher order Ising models were computed using CPLEX \cite{cplexv12, heavy_hex_QAOA_parameter_transfer2023}. 

As before in Sec.~\ref{sec:maxcut_40_rr3}, we only have a small number of unique layers of CNOT gates.
Since we want to cover a graph of degree three, we need at least three layers, see Appendix~\ref{sec:127_qubit_circuits}, with 144 CNOT gates in total.
The measured LF for the three layers is $LF_1 = 0.056926$, $LF_2 = 0.029630$ and $LF_3 = 0.167959$. 
These fidelities are significantly smaller than for the 40 qubit circuits in Sec.~\ref{sec:maxcut_40_rr3}.
The reason is that the qubits and gates on a 127-qubit devices are not all the same, there are always some better and some worse.
For 40 qubits, we could select the best line of 40 qubits (see Appendix~\ref{sec:40_qubit_circuits}), while for 127-qubits we had to use the whole chip.
From this we can again compute CNOT fidelity $\mathcal{F}_{CX} = (LF_1 \times LF_2 \times LF_3)^{1/144} = (0.000283)^{1/144} = 0.944850$, $EPLG = 0.055150$, and $\gamma_{CX} = 1.120146$. The results for evaluating the circuit on $ibm\_sherbrooke$, each with $10^5$ shots, are provided in Fig.~\ref{fig:127-qubit_results} and Tab.~\ref{tab:127_qubit_results}.
With the significantly lower fidelities, the number of shots required to apply the analytic CVaR bounds are significantly higher and currently impractical to run.
However, like in Sec.~\ref{sec:maxcut_40_rr3}, we see that the \emph{effective} $\gamma_{CX}$ is significantly smaller, even smaller than for the longer 40-qubit circuits. Further, we see that the noisy expectation values are still improving from $p=1$ until $p=4$ and only are starting to get worse for $p=5$.

Last, we use bootstrapping to confirm the scaling of the CVaR variance with respect to $\alpha$. More precisely, we uniformly sample $10^5$ values from the results collected using \emph{ibm\_sherbrooke} and estimate the CVaR for the five values of $\alpha_p'$ reported in Tab.~\ref{tab:127_qubit_results}. We repeat this $10^4$ times to estimate the variance of the resulting CVaR estimators. The results are provided in Fig.~\ref{fig:alpha_bootstrapping} and show close agreement with the theory presented in Sec.~\ref{sec:cvar}.

\begin{figure}
    \centering
    \includegraphics[width=\columnwidth]{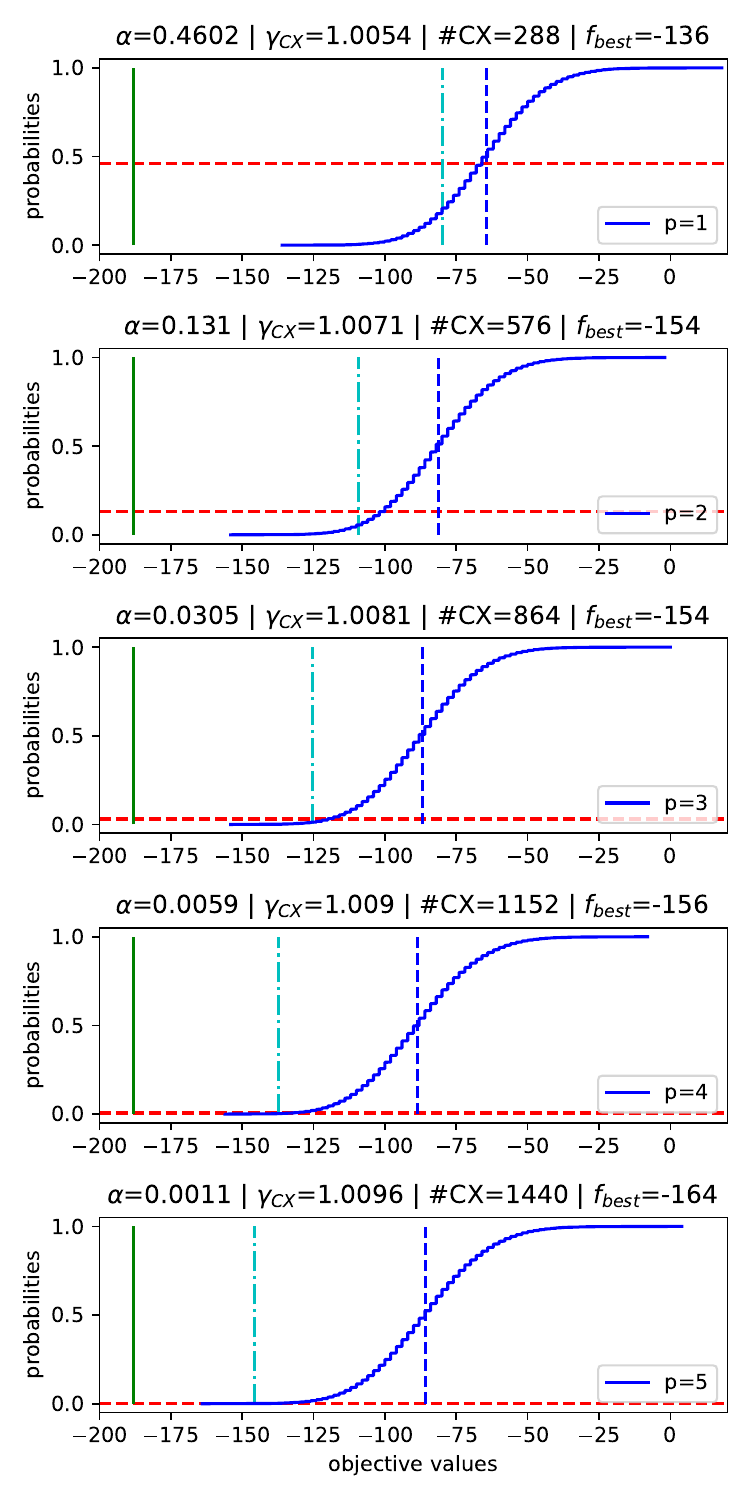}
    \caption{
    QAOA results for sampling a random hardware-compatible higher-order Ising model (minimization combinatorial optimization problem) on 127-qubits: This figure shows the resulting distributions from 127-qubit circuits executed on \emph{ibm\_sherbrooke} for $p=1,\ldots,5$ (top to bottom). The cumulative distribution functions show the values of the resulting samples from $10^5$ shots for every $p$. The vertical lines show the corresponding noisy expectation values (dashed blue), the noise-free expectation values evaluated using MPS simulation (cyan dashed-dotted), and the globally optimal solution equal to $-188$ (green solid). The title shows the fitted $\alpha_p'$ such that the $\CVaR_{\alpha_p'}$ are equal to the noise-free expectation values (i.e.~cyan dashed-dotted). The corresponding $\alpha_p'$ are indicated by the horizontal dashed red line.}    
    \label{fig:127-qubit_results}
\end{figure}

\begin{table*}
    \centering
    \begin{tabular}{lccccccccc}
        $p$ & $\#\text{CNOT}$ & $\tr(\rho H))$ & $\tr(\widetilde{\rho} H))$ & $f_{\text{best}}$ & $\sqrt{\gamma_p}$ & $\alpha_p$  & $\alpha_p'$ & $\gamma_{CX,p}'$ \\
        \hline
        1          &  288  &  -79.79 & -64.54 & -136 & $1.246 \times 10^{07}$  & $8.026 \times 10^{-08}$ & $0.4602$ & $1.0054$ \\
        2          &  576  & -109.35 & -81.11 & -154 & $1.553 \times 10^{14}$ & $6.441 \times 10^{-15}$ & $0.1310$ & $1.0071$ \\
        3          &  864  & -125.37 & -86.97 & -154 & $1.935 \times 10^{21}$ & $5.169 \times 10^{-22}$ & $0.0305$ & $1.0081$ \\
        4          & 1152  & -137.22 & -88.46 & -156 & $2.410 \times 10^{28}$ & $4.149 \times 10^{-29}$ & $0.0059$ & $1.0090$ \\
        5          & 1440  & -145.54 & -85.78 & -164 & $3.003 \times 10^{35}$ & $3.330 \times 10^{-36}$ & $0.0011$ & $1.0096$
    \end{tabular}
    \caption{QAOA results on 127-qubits: This table shows the different results for $p=1,\ldots,5$ when running QAOA on the introduced 127-qubit spin glass instance. It shows the number of CNOT gates per circuit, the noise-free and noisy expectation values, the best sampled values. Further, it shows the overall $\sqrt{\gamma_p}$ for the circuits and corresponding $\alpha_p$ derived from the LF as well as the $\gamma_{CX,p}'$ and $\alpha_p'$ derived from calibrating the CVaR on the noise-free expectation values.}
    \label{tab:127_qubit_results}
\end{table*}

\begin{figure}
    \centering
    \includegraphics[width=\columnwidth]{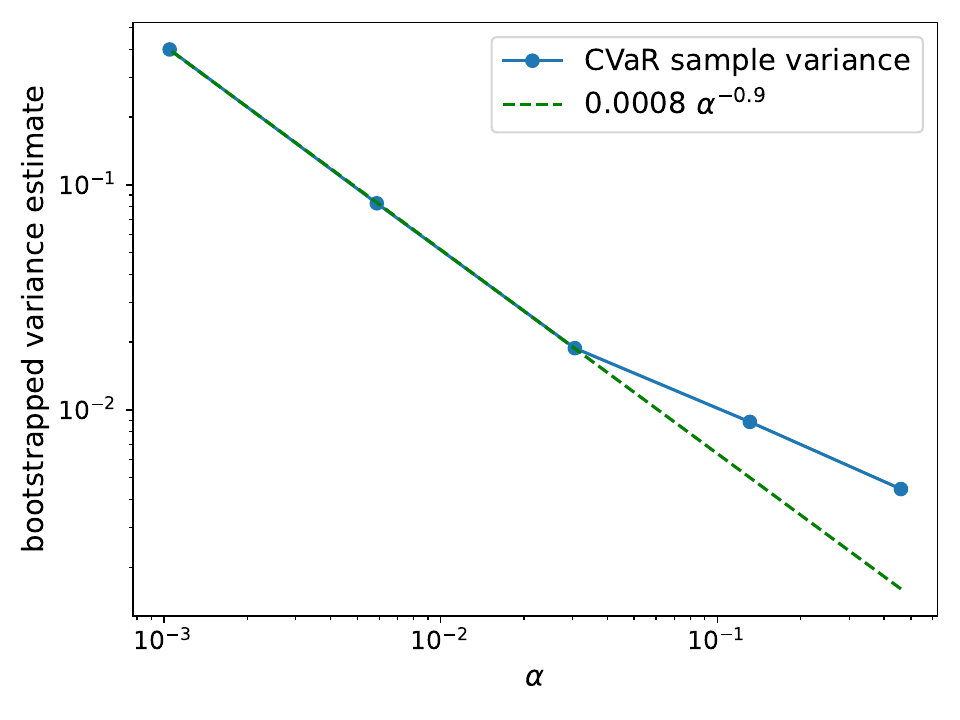}
    \caption{Variance of CVaR estimates: We draw $10^5$ uniform samples from the original data to estimate the CVaR for $\alpha_p'$, $p=1, \ldots 5$, cf.~Tab.~\ref{tab:127_qubit_results}, and repeat this $10^4$ times to get an estimate of the variance of the CVaR estimator. The dashed green line is fitted to the results for $p=3, \ldots, 5$, and is very close to the predicted behavior of $\mathcal{O}(1/\alpha)$.}
    \label{fig:alpha_bootstrapping}
\end{figure}

\section{Conclusion
\label{sec:conclusion}}

We examined how hardware noise affects the quality of bit strings sampled from quantum circuits on noisy quantum computers. We proved and demonstrated that the noise can be compensated by increasing the number of samples inversely proportional to the circuit's layer fidelity, or equivalently, proportional to $\sqrt{\gamma}$.
This is considerably less than that required for error mitigation strategies like probabilistic error cancellation, which scales as $\gamma^2$, however, to achieve unbiased estimators of expectation values instead of bounds.
Furthermore, we proved that the Conditional Value at Risk provides bounds on noise-free expectation values using noisy samples, providing the theoretical foundation for CVaR as a loss function in variational algorithms, and thus, closing a gap in the literature.
We also discussed the potential of this theory to benefit other algorithms, such as Quantum Support Vector Machines or approximate Variational Quantum Time Evolution.

Our primary focus was on errors occurring during circuit execution. However, other error sources, notably State Preparation and Measurement (SPAM) errors, also affect performance on noisy devices. The methodologies developed in this paper can be adapted to account for SPAM errors, either by increasing sampling overhead or applying other mitigation techniques, like statistical readout error mitigation.
The latter may allow to mitigate certain errors without added sampling overhead but might require additional calibration circuits. Investigating the impact of SPAM errors remains an intriguing direction for future research.
\\
\\
\noindent\textbf{Acknowledgments.}
The authors want to thank Almudena Carrera Vazquez, Julien Gacon, Youngseok Kim, David McKay, Diego Rist\`e, David Sutter, Kristan Temme, Minh Tran, and James Wootton for insightful discussions and recommendations to improve the theoretical and experimental results as well as the whole manuscript. Further, M.L.~and S.W.~acknowledge the support of the Swiss National Science Foundation, SNF grant No.~214919.
E.P., A.B.,~and S.E ~acknowledge the support of (i) the Beyond Moore's Law thrust of the Advanced Simulation and Computing Program (NNSA ASC) at Los Alamos National Laboratory (LANL), which is operated by Triad National Security, LLC, for the National Nuclear Security Administration of U.S. Department of Energy (Contract No. 89233218CNA000001), and (ii) LANL's Institutional Computing program. LANL report LA-UR-23-33295.

\appendix

\section{Assumption of Pauli noise}
\label{sec:beyond_pauli}

Within the theory of the paper we made the simplifying assumption of Pauli noise. This assumption is not given in general.
Suppose a Clifford quantum circuit layer $\mathcal{U}(\cdot) = U \cdot U^{\dagger}$ on $n$ qubits and its noisy version $\widetilde{\mathcal{U}} = \mathcal{U} \circ \Lambda$. A more realistic description of the noise is given by
\begin{eqnarray}
    \Lambda(\rho) &=& \sum_i A_i \rho A_i^{\dagger}\, , \label{eq:general_noise}
\end{eqnarray}
where the $A_i$ are Kraus operators \cite{nielsen_and_chuang}, which leads to
\begin{eqnarray}
    \widetilde{\mathcal{U}}(\rho) &=& \sum_i A_i U \rho U^{\dagger} A_i^{\dagger}\, . \label{eq:kraus_noise}
\end{eqnarray}
Applying Pauli twirling \cite{knill_randomized_2008, dankert_exact_2009, magesan_scalable_2011}, i.e., averaging over $\mathcal{U}$ conjugated by each element of the Pauli group on $n$ qubits yields
\begin{eqnarray}
    \widetilde{\mathcal{U}}_{\text{twirled}}(\rho) &=& 
    \frac{1}{4^n} \sum_{i, j} Q_j A_i U P_j \rho P_j U^{\dagger} A_i^{\dagger} Q_j\, , \label{eq:kraus_noise_U}
\end{eqnarray}
for Paulis $P_j, Q_j$ with $Q_j U P_j = U$ for all $j=1, \ldots, 4^n$.
This is known to translate the more general noise given in \eqref{eq:kraus_noise} on average to a Pauli noise model as given in \eqref{eq:noise_model}.
In practice, we do not enumerate all $4^n$ Paulis, but uniformly sample from them and apply a certain number of random Paulis to approximate the average.

Suppose now we have a noise model that-on average-looks like Pauli noise. Then, expectation values $\tr(\rho H)$ will have the same value in case of a true Pauli noise model as well as in case of a twirled general model.
That also holds if we set $H = \proj{x}$, i.e., we evaluate the probability of sampling $\ket{x}$. However, if we estimate the same sampling probability for the actual Pauli noise model and the twirled noise model the sampling probabilities also must be the same. 

For the experiments in Sec.~\ref{sec:experiments}, we omitted twirling.
There are some special cases of noise models where we know the theory holds exactly the same.
For instance, suppose \emph{stochastic noise} \cite{wallman2016bounding} $A_0 \sim \sqrt{I}$ and all other $A_i$, for $i > 0$ are orthogonal to $A_0$. Then, it can be easily seen that the probability of having no error is equal to the probability of the Pauli noise resulting after twirling, i.e., equal to $1/\sqrt{\gamma}$.
While we can always construct a noise channel with all orthogonal Kraus operators, it is not guaranteed that the identity is part of it.
In general, we can only say that the probability of no error in the general noise model is less than or equal to $1/\sqrt{\gamma}$ \cite{Wallman_2016, wallman2016bounding}.

However, it seems that the gap between the twirled and untwirled circuits is very small in the considered cases. We demonstrate the this by comparing the twirled and untwirled cases by comparing the resulting distributions. In Fig.~\ref{fig:twirl-comp} we show the experimental distributions when sampling from the \emph{ibm\_sherbrooke} device the same 127-qubits circuits discussed in Sec.~\ref{sec:127_poly3}. This shows a close agreement with and without twirling.

We note that the observed distributions in Fig.~\ref{fig:twirl-comp} deviate slightly from those presented in Fig.~\ref{fig:127-qubit_results}. This is because in order to twirl the circuits, we need to insert additional single qubit gates, which contribute to a slightly deeper circuit, here, about 8\% longer in the pulse schedule duration than the original circuits. In some cases this could be reduced by combining the twirling gates with other single qubit gates. However, if the additional gates are inserted, e.g., in between two CNOT gates, this is not possible. The circuits for the untwirled case have the same structure as the twirled case, except that the sampled twirling gates are constant, so that there is a fair comparison between the two due to the additional circuit duration.

We also note that the minimum values of the objective functions for the twirled case are lower than the untwirled case. However, since the opposite is true for the mean value of the objective function, we believe this may be due to sampling statistics, as in each of these cases the minimum objective value was only sampled only once.
If we determine $\alpha_p'$ as before for each case, we find that the twirled and untwirled values agree well for each $p$, and are well within a standard deviation of each other (determined by bootstrapping the observed bitstrings). This is summarized in Tab.~\ref{table:obj-vals-127}.

\begin{figure}[!ht]
    \centering
    \includegraphics[width=\columnwidth]{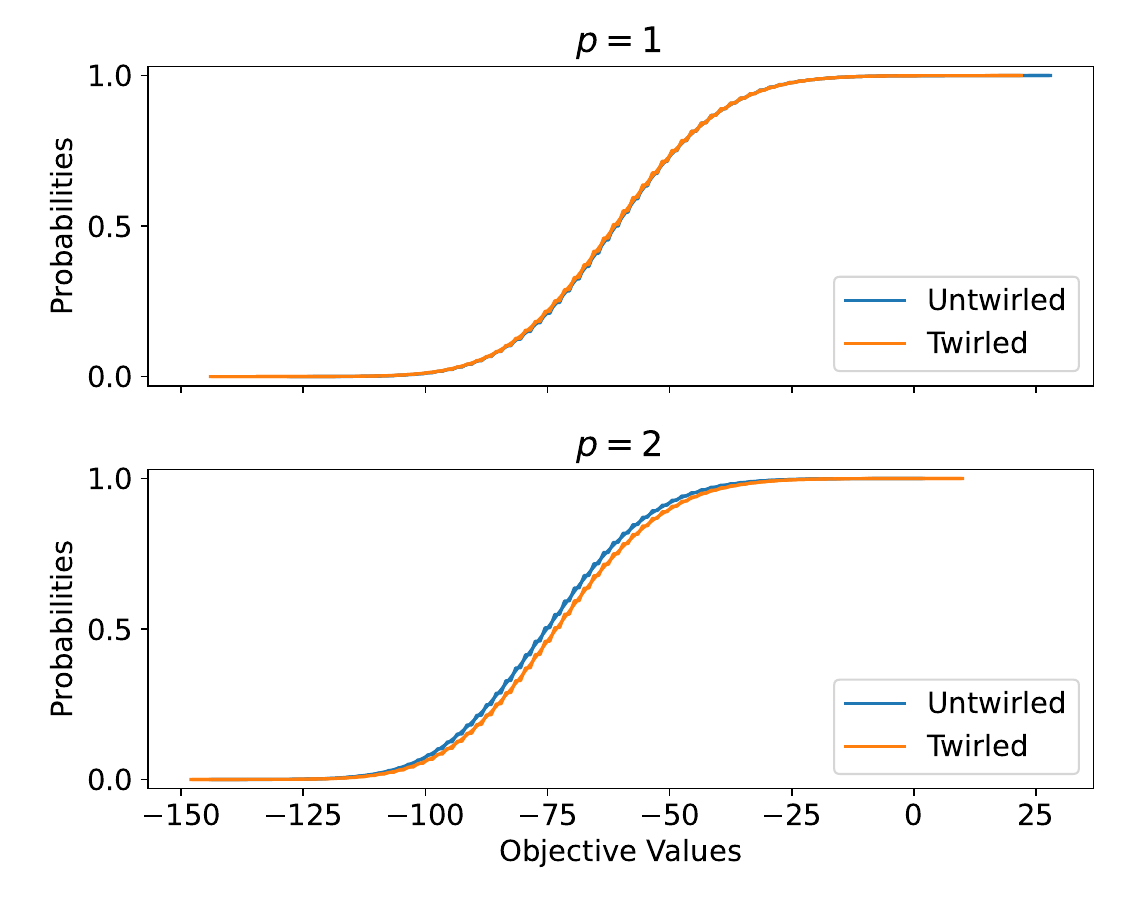}
    \caption{QAOA results on 127 qubits run on \emph{ibm\_sherbrooke}. This figure shows the cumulative distribution functions (CDFs) for depths $p=1,2$ with and without twirling. The orange lines correspond to the twirled circuits, and the blue lines correspond to the untwirled circuits. For each method, $100,000$ shots were used in total. When twirling, we sampled $1,000$ random twirls and performed $100$ shots for each. The statistics of these distributions are summarized in Table~\ref{table:obj-vals-127}.}
    \label{fig:twirl-comp}
\end{figure}

\begin{table}[h!]\label{table:obj-vals-127}
\centering
\begin{tabular}{clccccc}
p & Twirling & $\text{tr}(\rho H)$ & $\text{tr}(\widetilde{\rho}H)$ & $f_{\text{best}}$ & $\alpha_p'$ \\
\hline
\midrule
\multirow[c]{2}{*}{1} & No  & \multirow[c]{2}{*}{-79.8}  & -60.8 & -128 & 0.147 \ (7.9\%) \\
                      & Yes &                            & -60.9 & -144 & 0.152 \ (7.7\%) \\
\hline
\multirow[c]{2}{*}{2} & No  & \multirow[c]{2}{*}{-109.4} & -74.9 & -144 & 0.0202 \ (22.4\%) \\
                      & Yes &                            & -72.9 & -148 & 0.0160 \ (25.7\%) \\
\bottomrule
\end{tabular}
\caption{Values of the objective function obtained with and without twirling for QAOA depths $p=1,2$, corresponding to the distributions shown in Fig.~\ref{fig:twirl-comp}. The noise-free values are the expectation values of the observable obtained using classical MPS simulations (rounded to one decimal place) \cite{heavy_hex_QAOA_parameter_transfer2023}. The standard deviations for $\alpha_p'$ (shown as percent of the nominal value) are determined by bootstrapping over the observed bitstrings. We note that the $\alpha_p'$ here are lower than those in Fig.~\ref{fig:127-qubit_results} for the same reason that the observed values differ, as described in Appendix~\ref{sec:beyond_pauli}. Nevertheless, the qualitative conclusions still hold, and the twirled and untwirled cases agree well.}
\end{table}

\section{Probabilistic Error Cancellation \& Sampling}
\label{sec:sampling_pec}

In this section, we discuss how applying PEC \cite{berg2023probabilistic} to quantum circuits affects the resulting sampling probabilities. PEC consists of two steps: learning the noise when running a quantum circuit on a particular quantum device, and then, mitigating the noise to get an unbiased estimator of an expectation value. Here, we assume we have learned the noise already and focus on the error mitigation.
Given a noise model $\Lambda$, PEC constructs a \emph{Quasiprobability Decomposition} (QPD) to implement the inverse noise by combining multiple weighted quantum circuits.

In a QPD, a quantum operation $\mathcal{U}$ is implemented as a linear combination of other (possibly noisy) operations $\mathcal{E}_i$, $i = 1,\ldots,M$,
\begin{eqnarray} 
    \mathcal{U}(\cdot) &=& \sum_{i=1}^M a_i \mathcal{E}_i(\cdot) \, ,
    \label{eq:qpd_definition}
\end{eqnarray}
where $a_i \in \mathbb{R}$, $\mathcal{U}(X) = U X U^\dagger$, $\mathcal{E}_i$ denote (noisy) operations, and $\sum_{i=1}^M a_i = 1$.
This has first been proposed in the context of error mitigation \cite{Temme_2017}, where $\mathcal{U}$ is assumed to be a noise-free operation and $\mathcal{E}_i$ are noisy operations that can be implemented on a noisy device. If this is being applied to multiple gates and qubits, the number of necessary operations $M$ explodes exponentially.
Thus, instead of enumerating all of them, one rewrites \eqref{eq:qpd_definition} as
\begin{eqnarray}
  \mathcal{U}(\cdot) &=& \gamma \sum_{i=1}^M p_i s_i \mathcal{E}_i(\cdot)\,, \label{eq:qpd_U}
\end{eqnarray}
where $\gamma = \|a\|_1 \geq 1$, $p_i = |a_i|/\gamma$, and $s_i = \text{sign}(a_i)$,
and samples from the probability distribution defined through $p_i$.
Suppose we are interested in estimating $\braket{H} = \tr(\mathcal{U}(\rho)H)$ for some initial state $\rho$ and observable $H$.
Then, we can use the QPD to write
\begin{eqnarray}
    \tr(\mathcal{U}(\rho)H) &=&
    \gamma \sum_{i=1}^M p_i s_i \tr(\mathcal{E}_i(\rho)H).
\end{eqnarray}
Thus, instead of enumerating all $M$ circuits, we can sample from $p_i$, and only evaluate the sampled circuits corresponding to $i$, to get an unbiased estimator for $\braket{H}$.
However, the variance of this estimation is amplified by $\gamma^2$, i.e., $\gamma^2$-times more samples are needed than for the original noise-free circuit to achieve an estimate of the same accuracy.
The sampling overhead $\gamma^2$ grows exponentially in the number of qubits and depth of the circuit, and thus, can be prohibitively large for circuits beyond a certain circuit size and noise levels.

While PEC has only been considered for the estimation of expectation values, it also generates samples from every random circuit that is measured. However, we will show that this essentially amplifies the noise and increases the sampling overhead compared to the results presented within this paper.
To this extent, we introduce the following mixed state introduced by PEC:
\begin{eqnarray}
    \rho_{\text{PEC}} &=& \sum_{i=1}^M p_i  \mathcal{E}_i(\rho), \label{eq:rho_pec}
\end{eqnarray}
for some initial state $\rho$.
The state $\rho_{\text{PEC}}$ is achieved by dropping the factor $\gamma$ as well as the signs $s_i$ from \eqref{eq:qpd_U}.
This allows us to state the following lemma.

\begin{lemma} \label{lemma:pec_sampling_prob}
Suppose a $n$-qubit state $\rho = U \proj{0} U^{\dagger}$, where $U$ is some unitary, with
\begin{eqnarray}
  \tr( \rho \proj{x} ) &=& p_x \geq 0, \label{eq:p_x}
\end{eqnarray}
for a computational basis state $\ket{x}$, $x\in\{0, 1\}^n$. 

Further, suppose that $U$ can be error-mitigated on a noisy device by using PEC with corresponding $\gamma \geq 1$ and denote the resulting mixed state introduced in \eqref{eq:rho_pec} by $\rho_{\text{PEC}}$.
Then, the probability of measuring $\ket{x}$ on the noisy devices using PEC is lower bounded by
\begin{eqnarray}
  \tr( \rho_{\text{PEC}} \proj{x} ) = p_x^{\text{PEC}}&\geq& p_x / \gamma. \label{eq:p_x_lb}
\end{eqnarray}
\end{lemma}

\begin{proof}
Consider the QPD resulting from PEC
\begin{align} \label{eq_QDP}
    \mathcal{U}(\cdot) = \sum_{i=1}^M a_i \mathcal{E}_i(\cdot) \, .
\end{align}
Using~\eqref{eq_QDP} we can write
\begin{align} \label{eq:pec_qpd}
p_x 
= \tr( \rho \proj{x} )
= \sum_{i=1}^M a_i \tr( \mathcal{E}_i(\proj{0})\proj{x})\, .
\end{align}
By defining $\gamma = \| a \|_1$, $p_i = |a_i|/\gamma$, and $s_i = \text{sign}(a_i)$, we can rewrite \eqref{eq:pec_qpd} as
\begin{align}
  p_x =  \gamma \sum_{i=1}^M p_i s_i \tr( \mathcal{E}_i(\proj{0})\proj{x})\,. \label{eq:pec_qpd_2}
\end{align}
Further, $s_i \tr( \mathcal{E}_i(\proj{0})\proj{x})$ allows us to define a random variable $Y_i \in \{-1, 0, +1\}$ that equals $\pm 1$ if we measure $\mathcal{E}_i(\proj{0})$ and obtain $\ket{x}$, where the sign is determined by $s_i$, and $0$ otherwise. The random variable $Y_i$ satisfies $\mathbb{E}[Y_i] = s_i \tr( \mathcal{E}_i(\proj{0}) \proj{x})$. 
We denote the probabilities of $Y_i$ taking the values $-1, 0, +1$ by $q_i^{-1}, q_i^0, q_i^{+1} \geq 0$, respectively.
Note that by construction, for each $i$ only one of $q_i^{-1}, q_i^{+1}$ can be larger than zero. 

In addition, let the probabilities $p_i$ define a random variable $I \in \{1, \ldots, M\}$. Then, by the law of total expectation, we get 
\begin{align}
\gamma \mathbb{E}[Y_I]
&= \gamma \sum_{i=1}^M \mathbb{E}[Y_i|i] \mathbb{P}[i] \\
&= \gamma \sum_{i=1}^M p_i s_i \tr( \mathcal{E}_i(\proj{0})\proj{x}) \\
&= p_x \, .
\end{align}
This can be rewritten as
\begin{eqnarray}
     \sum_{i=1}^M p_i \left( q_i^{+1} - q_i^{-1} \right) &=& \frac{p_x}{\gamma}\,. \label{eq:p_x_pec}
\end{eqnarray}

The total probability to measure $\ket{x}$ when applying PEC, independent of the sign of $Y_I$, is then given by
\begin{eqnarray}\label{eqn:lower_bound}
     \sum_{i=1}^M p_i \left( q_i^{+1} + q_i^{-1} \right) &\geq& \frac{p_x}{\gamma}, \label{eq:p_x_pex_lb}
\end{eqnarray}
where the lower bound follows immediately from \eqref{eq:p_x_pec}, and the right-hand-side is exactly the probability of measuring $\ket{x}$ for state $\rho_{\text{PEC}}$.
\end{proof}

If we compare the result from Lemma~\ref{lemma:pec_sampling_prob} with the lower bound presented in \eqref{eq:prob_lower_bound}, we see that PEC implies the squared overhead compared to direct sampling. Further, this implies that CVaR-based approaches may significantly reduce the overhead to achieve insightful results, particularly when combined with problem structure to filter noisy samples.

\section{Variance of Estimating the CVaR}
\label{sec:cvar_variance}

In this section, we present a short exposition on how to estimate CVaR. We will first state the following lemma.

\begin{lemma}\label{lemma:cvar_sampling}
    Let $X_1,\dots,X_n$ be i.i.d.~ copies of $X$ (with $X$ integrable) and let $X_{(1)},\dots,X_{(n)}$ be their order statistic. For $\alpha \in (0,1]$ let $E_n = (X_{(1)}+\cdots+X_{(\lfloor \alpha n\rfloor)})/\lfloor \alpha n\rfloor$. Then
    \begin{align*}
        \mathbb{E}[E_n] \to \CVaR_\alpha(X) \quad \text{as }n\to \infty \, .
    \end{align*}
    If $X$ is square integrable and $F_X(x_\alpha)=\alpha$,
    \begin{align*}
        \sqrt{n} (E_n - \CVaR_\alpha(X))
        \to N(0,\CVaRv_\alpha(X))
    \end{align*}
    in distribution as $n\to \infty$ where here $\CVaRv_\alpha(X) := \alpha^{-1} \Var[X \mid X \le x_\alpha]$ is the limiting variance.
\end{lemma}

To estimate $\overline{\CVaR}_\alpha(X)$, we use the estimator $\overline{E}_n = (X_{(n-\lfloor \alpha n\rfloor+1)}+\cdots+X_{(n)})/\lfloor \alpha n\rfloor$ and obtain analogous results.

\begin{proof}
    Recall $F_X(x) = \mathbb{P}[X \leq x]$ and define $F_X(x-) = \mathbb{P}[X < x]$. We make the following definitions for (left limits) of empirical cumulative distribution functions:
    \begin{align*}
        \hat{F}_n(x) &= \#\{ i \le n \colon X_i \le x\}/n \,,\\
        \hat{F}_n(x-) &= \#\{ i \le n \colon X_i < x\}/n \,.
    \end{align*}
    Also let $\Delta F_X(x) = F_X(x) - F_X(x-)$ and $\Delta \hat{F}_n(x) = \hat{F}_n(x) - \hat{F}_n(x-)$. The key observation is that
    \begin{align*}
        E_n = \frac{1}{\lfloor \alpha n\rfloor} \sum_{i=1}^n X_i \min\left\{\frac{(\lfloor \alpha n\rfloor - n\hat{F}_n(X_i-))_+}{n\Delta \hat{F}_n(X_i)},1 \right\} \,.
    \end{align*}
    Indeed, any $x\in \mathbb{R}$ will appear in the sum defining $\lfloor \alpha n\rfloor E_n$ precisely $\min\{(\lfloor \alpha n\rfloor - n\hat{F}_n(x-))_+, n\Delta \hat{F}_n(x)\}$ times; the $n\Delta \hat{F}_n(X_i)$ in the denominator above takes care of overcounting. Now
    \begin{align*}
        &\mathbb{E}[E_n] = \frac{n}{\lfloor \alpha n\rfloor} \mathbb{E} \left[X_1 \min\left\{\frac{(\lfloor \alpha n\rfloor - n\hat{F}_n(X_1-))_+}{n\Delta \hat{F}_n(X_1)},1 \right\} \right] \\
        &= \frac{n}{\lfloor \alpha n\rfloor} \mathbb{E} \left[A_n(X_1) \right]
    \end{align*}
    where
    \begin{align*}
        A_n(x) := x \mathbb{E} \left[\min\left\{\frac{(\lfloor \alpha n\rfloor - n\hat{F}_{n-1}(x-))_+}{1+n\Delta \hat{F}_{n-1}(x)},1 \right\} \right] \,.
    \end{align*}
    The first equality above follows from the linearity of the expectation and the i.i.d.~property of $X_1,\dots,X_n$ and the second equality follows from conditioning on $X_1$. Using the strong law of large numbers we have $(\hat{F}_n(x),\hat{F}_n(x-),\Delta \hat{F}_n(x)) \to (F_X(x),F_X(x-),\Delta F_X(x))$ a.s.\ as $n\to \infty$. By separately considering the $\Delta F_X(x) = 0$ and $\Delta F_X(x) > 0$ cases we get
    \begin{align*}
        A_n(x) &\to x \cdot 1(F_X(x) < \alpha) \\
        &\qquad + x\frac{\alpha - F_X(x-)}{\Delta F_X(x)} \,1(\alpha \in (F_X(x-), F_X(x)))
    \end{align*}
    as $n\to \infty$ unless $\alpha = F_X(x)=F_X(x-)$; however we have $\mathbb{P}[\alpha = F_X(X_1)=F_X(X_1-)]=0$ so this case does not matter to evaluate the limit of $\mathbb{E}[E_n]$. Thus by dominated convergence
    \begin{align*}
        \mathbb{E}[E_n] &\to \alpha^{-1} \mathbb{E}[X_1 ; F_X(X_1) < \alpha] \\
        &\qquad + \sum_{x\colon \alpha \in (F_X(x-), F_X(x))} x(1 - \alpha^{-1} F_X(x-))\\
        &= \CVaR_\alpha(X)
    \end{align*}
    as $n\to \infty$. The second claim on the central limit theorem is a special case of \cite{cvarestimator_clt}.
\end{proof}

Let us make the following remark on monotonicity: If $\phi\colon \mathbb{R}\to \mathbb{R}$ is non-decreasing and $\phi(X)$ is integrable, then
\begin{align*}
    0 &\geq \mathbb{E}[(\phi(X)-\phi(X'))(1(X\leq x) - 1(X'\leq x))] \\
    &= 2\cdot \mathbb{E}[\phi(X);X\leq x] - 2\cdot \mathbb{E}[\phi(X)] \mathbb{P}[X\leq x] \,.
\end{align*}
By applying this to $\phi(x) = x$ and $x=x_{\alpha}$ we see that $\CVaR_\alpha(X)\leq \mathbb{E}[X]$. Furthermore, by replacing $X$ by a random variable sampled from the law of $X$ conditioned on $X\leq x_{\alpha'}$ for $\alpha'>\alpha$ we can deduce that $\CVaR_\alpha(X)$ is non-decreasing in $\alpha$. Much more crudely, we can bound $\CVaRv_\alpha(X)\leq \alpha^{-1}\mathbb{E}[X^2]/\mathbb{P}[X\leq x_\alpha] \leq \mathbb{E}[X^2]/\alpha^2$.

In the following, we analyze behavior of the limiting distribution of the estimator $E_n$ in some concrete cases.

In the case where $X$ has a Bernoulli distribution with success probability $p$, we observe that $\overline{E}_n$ has the same distribution as $\min\{B_n/\lfloor \alpha n\rfloor,1\}$ where $B_n$ is Binomial distributed with parameter $(n,p)$. An application of the central limit theorem thus yields
\begin{align*}
    &\sqrt{n}(E_n - \min\{p/\alpha, 1\}) \\
    &\to \begin{cases}
        \alpha^{-1}\sqrt{p(1-p} N &\colon \alpha > p\,\\
        \sqrt{(1-p)p^{-1}}\,N\cdot 1(N\geq 0) &\colon \alpha = p,\\
        0 &\colon \alpha < p\,.
    \end{cases}
\end{align*}
in distribution as $n\to \infty$ where $N$ is a standard normal random variable.

To analyze the case where $N\sim N(0,1)$, it will be useful to recall the following asymptotic expansion \cite[(8.11(i))]{nist_dlmf} of incomplete Gamma functions:
\begin{align*}
    \Gamma(a, y) &:= \int_y^\infty s^{a-1} e^{-s} \,ds \\
    &= y^{a-1}e^{-y} \left( \sum_{k=0}^{n-1} \frac{(a-1)\cdots (a-k)}{y^k} + O(y^{-n})\right)
\end{align*}
as $y \to \infty$ for any fixed $n \ge 1$ and $a>0$. In particular as $x \to \infty$,
\begin{align*}
    \frac{\Gamma(1/2,x^2/2)}{\sqrt{2}} &= \frac{e^{-x^2/2}}{x} \left( 1 - \frac{1}{x^2} + \frac{3}{x^4} + O(x^{-6})\right),\\
    \frac{\sqrt{2}}{\Gamma(1/2,x^2/2)} &= x e^{x^2/2} \left( 1 + \frac{1}{x^2} - \frac{2}{x^4} + O(x^{-6})\right).
\end{align*}
Let $x_\alpha = F_N^{-1}(\alpha)$ and write $f_N=F_N'$ for the density of $N$. By \cite[(7.17(iii))]{nist_dlmf} we get the asymptotic relationship
\begin{align*}
    x_\alpha \sim - \sqrt{-\log(4\pi \alpha^2 \log(1/(2\alpha)))}
\end{align*}
as $\alpha \to 0$. We will compute $\CVaR_\alpha(N)$ and $\CVaRv_\alpha(N)$ via the cumulant generating function $\phi$ of a truncated Gaussian
\begin{align*}
    \phi(\theta) &= \log \mathbb{E}[e^{\theta N} \mid N \le x_\alpha] \\
    &= - \log F_N(x_\alpha)
    + \log \int_{-\infty}^{x_\alpha} \frac{e^{-t^2/2 + \theta t}}{\sqrt{2\pi}}\, dt \\
    &= - \log F_N(x_\alpha)
    + \log\int_{-\infty}^{x_\alpha} \frac{e^{-(t-\theta)^2 /2 + \theta^2/2}}{\sqrt{2\pi}} \, dt \\
    &= - \log F_N(x_\alpha) + \theta^2/2 + \log F_N(x_\alpha-\theta) .
\end{align*}
Differentiating at $\theta = 0$ yields the expressions
\begin{align*}
    \CVaR_\alpha(N) &= \phi'(0) = -\frac{1}{\alpha} f_N(x_\alpha) ,\\
    \alpha \CVaRv_\alpha(N) &= \phi''(0) = 1 + \frac{1}{\alpha} f_N'(x_\alpha) - \frac{1}{\alpha^2} f_N(x_\alpha)^2 \\
    &= 1 - \frac{1}{\alpha} x_\alpha f_N(x_\alpha) - \frac{1}{\alpha^2} f_N(x_\alpha)^2 .
\end{align*}
Since $\alpha = F_N(x_\alpha) = \Gamma(1/2,x_\alpha^2/2)/(2\sqrt{\pi})$, it follows that as $\alpha \to 0$,
\begin{align*}
    \CVaR_\alpha(N) &= x_\alpha \left( 1 + \frac{1}{x_\alpha^2} - \frac{2}{x_\alpha^2} + O(x_\alpha^{-6}) \right), \\
    \CVaRv_\alpha(N) &= \frac{1}{\alpha} + \frac{x_\alpha^2}{\alpha} \left( 1 + \frac{1}{x_\alpha^2} - \frac{2}{x_\alpha^2} + O(x_\alpha^{-6}) \right) \\
    &\qquad - \frac{x_\alpha^2}{\alpha} \left( 1 + \frac{1}{x_\alpha^2} - \frac{2}{x_\alpha^2} + O(x_\alpha^{-6})\right)^2 \\
    &= \frac{1}{\alpha x_\alpha^2} + O(x_\alpha^{-4}).
\end{align*}
As a final example, we can consider the case where $X$ has density $f_X(x)=\beta x^{-1-\beta}1(x\geq 1)$ where $\beta > 0$ (i.e.,~we consider a power law tail). Here, one can compute that for $\beta>2$, $\overline{\CVaRv}_\alpha(X)=\beta(\beta-1)^{-2}(\beta-2)^{-1}\alpha^{-2/\beta}$ which is worse than the decay in the standard normal case and achieves the worst case upper bound on the variance in the $\beta\to 2$ limit.

\section{Relation to brute-force search}
\label{sec:brute_force}

A brute-force search enumerates all $2^n$ candidate solutions and checks which one is optimal.
The sampling overhead of $\sqrt{\gamma}$ on noisy devices can thus be related to brute-force search thereby allowing us to derive a hardware requirements for QAOA.
Assuming, for simplicity, that the probability $p_x$ to sample the optimal solution is close to $1$ we require hardware with $\sqrt{\gamma} < 2^n$.
We can relate this to the layer fidelity to obtain a requirement on hardware quality necessary for potential quantum advantage.
First, we assume that each layer $i$ in a QAOA circuit has the same layer fidelity ${\rm LF}=1/\sqrt{\gamma_i}$. 
As a result the $\gamma$ of the circuit is $\gamma=\prod_{i=1}^{d(n)}\gamma_i=1/{\rm LF}^{2d(n)}$ where $d(n)$ is the depth defined as the number of non-overlapping two-qubit gate layers.
This assumption is reasonable when transpiling QAOA circuits to a line of qubits which requires layers of CNOT gates applied on every other edge~\cite{weidenfeller2022scaling}.
Therefore, the sampling cost to compensate for noise is $1/{\rm LF}^{d(n)}$.
For a line of qubits we may assume that to leading order $d(n)\sim 3n p$.
The factor $3n$ comes from the fact that $n-2$ layers of SWAP gates are needed to implement full connectivity and each SWAP merged with an $R_{ZZ}$ is implemented with three CNOT gates.
Here, $p$ is the number of QAOA layers which is sometimes assumed to grow with the logarithm of problem size, i.e., $p\propto\log(n)$ \cite{Bravyi2019, weidenfeller2022scaling}.
If the sampling overhead should stay below brute-force search we therefore require ${\rm LF}^{-3np}<2^n$ which implies that the layer fidelity must satisfy
\begin{align}
    {\rm LF} > \frac{1}{2^{1/3p}}.
\end{align}
This requirement is only dependent on problem size through the relation between $p$ and $n$.
However, as shown in Ref.~\cite{mckay2023benchmarking} the layer fidelity decreases with the number of qubits in the layer.
If we further assume that layers are dense, i.e., every layer on $n$ qubits consists of approximately $n/2$ CNOT gates, we can compute a corresponding CNOT fidelity as ${\rm LF}^{2/n}$, as well as the corresponding lower bound
\begin{align}
    {\rm LF}^{2/n} > \frac{1}{2^{2/3pn}}.
\end{align}

\section{40-qubit Circuits}
\label{sec:40_qubit_circuits}

The 40-qubit circuits in the main texts are based on those in Ref.~\cite{Sack2023}.
In this work, the authors consider random three-regular graphs transpiled to a line of qubits using a swap network~\cite{weidenfeller2022scaling}.
This results in circuits that alternate only two types of layers of CNOT gates as described in the main text.
Furthermore, the authors carefully chose the decision variable to physical qubit mapping to minimize the number of layers of the swap network.
This method is described in Ref.~\cite{Matsuo2023}.
The code to produce such circuits is available on GitHub~\cite{BestPractices}.
The optimal parameters resulting from the light-cone optimization are given by
$(\gamma_1, \beta_1) = (2.8405, 0.3982)$ for $p=1$ and 
$(\gamma_1, \beta_1, \gamma_2, \beta_2) = (1.1506, 0.3288, 0.1941, 0.6582)$ 
for $p=2$, respectively.

\section{Dynamical Decoupling}
\label{sec:staggered_dd}

Dynamical decoupling (DD) removes an interaction between a system and a bath by inserting pulses~\cite{Viola1998, Zanardi1999, Vitali1999}.
Here, we briefly summarize DD following Ref.~\cite{Ezzell2023}.
Consider a time-independent bath $H_B$ interacting with the system $H_S=H_S^0+H_S^1$ though $H_{SB}$.
Here, $H_S^1$ is an undesired, always-on error term.
The goal of DD is to insert pulses in idle times such the time evolution of the system and bath becomes $U(T)=U_0(T)B(T)$ with $U_0=\exp(-iTH_S^0)\otimes \mathbb{I}_B$ the desired error-free time-evolution and $B(T)$ ideally acts on the bath alone.

Consider a single qubit with $H_B+H_{SB}=\sum_{\alpha=0}^3 \gamma_\alpha\sigma^\alpha\otimes B^\alpha$.
Here, $\gamma_\alpha$ is a coefficient, and $B^\alpha$ is the bath term that couples to the qubit through the  $\sigma^\alpha$ Pauli matrix.
The simplest DD sequence is ${\rm PX}=X - d_\tau  - X - d_\tau$ where $d_\tau$ indicates a delay of duration $\tau$.
Since $X$ anti-commutes with $Y$ and $Z$, the sequence ${\rm PX}$ cancels the $Y\otimes B^Y$ and $Z\otimes B^Z$ system-bath interactions.
The effective error Hamiltonian after a duration $2\tau$ is $H^\text{err}_{\rm PX}=\gamma_x X\otimes B^x+\mathbb{I}_s\otimes \tilde B+\mathcal{O}(\tau^2)$.
Here, we see that $\rm PX$ is not universal since an $X$ error remains. 
Universal decoupling up to first-order is achieved with the $XY4$ sequence
\begin{align}
    {\rm XY4} = Y - d_\tau  - X - d_\tau - Y - d_\tau  - X - d_\tau
\end{align}
which results in the effective error Hamiltonian $H^\text{err}_{\rm XY4}=\mathbb{I}_s\otimes \tilde B+\mathcal{O}(\tau^2)$.

We now consider the two-qubit case.
Two fixed-frequency qubits typically exhibit an undesired $ZZ$-coupling which is effectively suppressed with DD~\cite{Tripathi2022, Mundada2023}.
Simultaneously applying the $\rm PX$ sequence on both qubits cancels unwanted errors arising from $\mathbb{I}\otimes Z$ and $Z\otimes \mathbb{I}$.
However, since simultaneous $X$ pulses commute with $Z\otimes Z$ the unwanted $ZZ$ interactions (i.e. cross-talk, which is common in transmon qubits) are still present.
This is remedied with staggered DD.
We apply the sequence
\begin{align}
    X_1 - d_{\tau} - X_0 - d_{\tau} - X_1 - d_{\tau} - X_0 - d_{\tau}
\end{align}
which staggers two $\rm PX$ sequences. 
Here, $X_i$ is an $X$ gate applied to qubit $i$, which inverts the evolution of $Z_i$ and $Z_1\otimes Z_0$.
In total, the evolution of single-qubit $Z_i$ errors changes sign twice and the evolution of $ZZ$ errors changes sign four times.
In this work we apply the staggered XY4 sequence \cite{zhou_quantum_2022} (a variant of the staggered XX sequence presented in \cite{Mundada2023}) to ensure a proper cancellation of two-qubit static cross-talk.
The staggered XY4 sequence we employ is defined by $Y_0 - d_{\tau} - Y_1 - d_{\tau} - X_0 - d_{\tau} - X_1 - d_{\tau} - Y_0 - d_{\tau} - Y_1 - d_{\tau} - X_0 - d_{\tau} - X_1 - d_{\tau}$.
As discussed above, it is universal for single-qubit terms and will cancel the static $ZZ$ cross-talk between qubits.

\section{127-qubit QAOA Circuits}
\label{sec:127_qubit_circuits}

\begin{figure*}[t!]
    \centering
    \includegraphics[width=2\columnwidth]{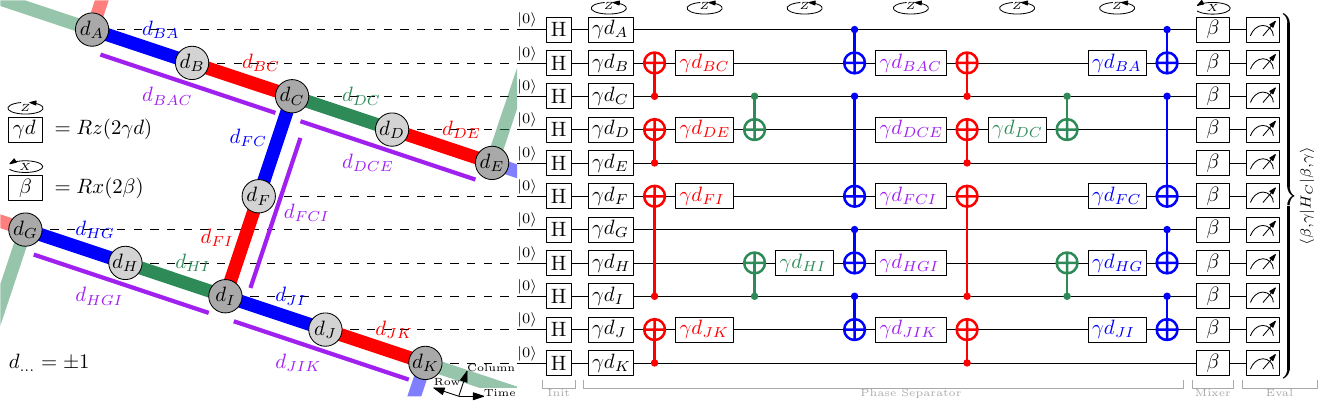}
    \caption{\textbf{From Refs.}~\cite{pelofske2023qavsqaoa,pelofske2023short}: Diagram of a heavy-hex graph compatible $p=1$ QAOA circuit for sampling heavy-hex compatible higher order Ising models (specifically cubic terms centered on degree 2 nodes). Left hand side of the figure shows the 3-edge-coloring and the bipartition of the graph, and the cubic terms are denoted by the adjacent purple lines next to the hardware graph. The right hand side of the figure shows the corresponding QAOA circuit for this sub-component of the heavy-hex graph (which can be extended arbitrarily to a large heavy-hex graph, and to higher $p$). The cubic terms are addressed by a single layer of Rx rotation gates (shown in purple), and the total CNOT depth per $p$ is always $6$. Following the phase separator, the transverse field mixer is applied and then the state of all qubits are measured after $p$ rounds have been applied. }
    \label{fig:heavy_hex_QAOA_circuit}
\end{figure*}

In Figure~\ref{fig:heavy_hex_QAOA_circuit}, taken from Refs.~\cite{pelofske2023qavsqaoa,pelofske2023short}, we briefly discuss the optimized circuits for the 127-qubit higher-order instances to have a self-contained description. This illustrates that all 2-qubit gates needed for the implementation of $e^{-i\gamma H}$ can be scheduled in just 3 different layers of non-overlapping CNOT gates. In each QAOA round $p$, each layer is used once to compute and once to uncompute $ZZ$ and $ZZZ$ parity values, for an overall CNOT depth of $6p$. The exact values of the heuristically computed, using parameter transfer, QAOA angles that give a strictly increasing expectation value as $p$ increases up to $5$ are given in Ref.~\cite{heavy_hex_QAOA_parameter_transfer2023}.

\bibliographystyle{arxiv_no_months}
\bibliography{references}

\end{document}